\documentclass{article} 
\usepackage{graphicx}

\usepackage{booktabs}
\usepackage[first=0,last=9]{lcg}
\usepackage{setspace}
\usepackage{amsmath}
\usepackage{amsfonts}
\usepackage{graphicx}
\usepackage{caption}
\usepackage{subcaption}
\usepackage{epstopdf}
\usepackage{grffile}
\usepackage{float}
\usepackage{multirow}
\usepackage{algorithm}
\usepackage{algpseudocode}
\usepackage{color, colortbl}

\newcommand{\vgam}{\boldsymbol{\gamma}}

\newcommand{\va}{\boldsymbol{\alpha}}

\newcommand{\range}{{\bf range}}

\newcommand{\Zr}{\mathbb{Z}}
\newcommand{\JE}{\widehat{\mathcal{J}}}

\newcommand{\vw}{{\bf w}}
\newcommand{\LM}{{\bf L}}

\definecolor{Gray}{gray}{0.9}

\newcommand{\dx}{\boldsymbol{\Delta}{\bf x}}

\newcommand{\dv}{{\bf d}}
\newcommand{\ra}{r}

\newcommand{\D}{{\bf D}}

\newcommand{\bal}{\boldsymbol{\alpha}}
\newcommand{\bv}{\boldsymbol{\beta}}

\newcommand{\DX}{\boldsymbol{\Delta}{\bf X}}

\newcommand{\ones}{{\bf 1}}

\newcommand{\xm}{\overline{\bf x}}
\newcommand{\M}{{\bf M}}

\newcommand{\N}{M}
\newcommand{\Q}{{\bf Q}}
\newcommand{\z}{{\bf z}}
\newcommand{\x}{{\bf x}}
\newcommand{\X}{{\bf X}}
\renewcommand{\Re}{\mathbb{R}}
\newcommand{\leftl}{\left \{}
\newcommand{\rightl}{\right \}}
\newcommand{\leftp}{\left (}
\newcommand{\rightp}{\right )}
\newcommand{\leftb}{\left [}
\newcommand{\rightb}{\right ]}
\newcommand{\lefta}{\left |}
\newcommand{\righta}{\right |}

\newcommand{\BO}[1]{\mathcal{O}\leftp #1 \rightp}
\newcommand{\BE}{\widehat{\bf B}}
\newcommand{\I}{{\bf I}}

\newcommand{\Nor}{\mathcal{N}}
\newcommand{\norm}[1]{\left \| #1 \right \|}
\newcommand{\J}{\mathcal{J}}
\newcommand{\JT}{\widehat{\mathcal{J}}}
\newcommand{\Ho}{\mathcal{H}}
\newcommand{\HG}{\widehat{\mathcal{H}}}

\newcommand{\erra}{\boldsymbol{\gamma}}
\newcommand{\var}{{\bf var}}
\newcommand{\evar}{\widehat{\bf var}}
\newcommand{\y}{{\bf y}}
\newcommand{\B}{{\bf B}}
\newcommand{\R}{{\bf R}}

\newcommand{\zero}{{\bf 0}}
\newcommand{\Mo}{\mathcal{M}}

\newcommand{\errobs}{\boldsymbol{\varepsilon}}

\newcommand{\Nstate}{n}
\newcommand{\Nobs}{m}
\newcommand{\Nens}{N}
\renewcommand{\P}{{\bf P}}
\renewcommand{\H}{{\bf H}}

\newcommand{\LE}{\widehat{\bf L}}
\newcommand{\DE}{\widehat{\bf D}}

\newcommand{\JN}{\widetilde{\J}}

\usepackage{amsthm}
\usepackage[affil-it]{authblk}

\newtheorem{theorem}{Theorem}[section]

\begin{document}

\title{Adjoint-Free 4D-Var Methods Via Line Search Optimization For Non-Linear Data Assimilation
}




\author[1]{Elias D. Nino-Ruiz%
  \thanks{Electronic address: \texttt{enino@uninorte.edu.co}; Corresponding author}}

\author[2]{J. Diaz-Rodriguez%
  \thanks{Electronic address: \texttt{jdiazrod@yorku.ca}}}

\affil[1]{Applied Math and Computer Science Laboratory\\ 
Department of Computer Science\\
Universidad del Norte, Colombia}

\affil[2]{Department of Mathematics and Statistics\\
York University, Canada}


\maketitle

\begin{abstract}
This paper proposes two practical implementations of Four-Dimensional Variational (4D-Var) Ensemble Kalman Filter (4D-EnKF) methods for non-linear data assimilation. Our formulations' main idea is to avoid the intrinsic need for adjoint models in the context of 4D-Var optimization and, even more, to handle non-linear observation operators during the assimilation of observations. The proposed methods work as follows: snapshots of an ensemble of model realizations are taken at observation times, these snapshots are employed to build control spaces onto which analysis increments can be estimated. Via the linearization of observation operators at observation times, a line-search based optimization method is proposed to estimate optimal analysis increments. The convergence of this method is theoretically proven as long as the dimension of control-spaces equals model one. In the first formulation, control spaces are given by full-rank square root approximations of background error covariance matrices via the Bickel and Levina precision matrix estimator. In this context, we propose an iterative Woodbury matrix formula to perform the optimization steps efficiently. The last formulation can be considered as an extension of the Maximum Likelihood Ensemble Filter to the 4D-Var context. This employs pseudo-square root approximations of prior error covariance matrices to build control spaces. Experimental tests are performed by using the Lorenz 96 model. The results reveal that, in terms of Root-Mean-Square-Error values,  both methods can obtain reasonable estimates of posterior error modes in the 4D-Var optimization problem. Moreover, the accuracies of the proposed filter implementations can be improved as the ensemble sizes are increased. 

\end{abstract}

\section{Introduction}
\label{sec:introduction}

To be concise, the goal of Data Assimilation (DA) is to correct imperfect model trajectories $\leftl \x^b_k \rightl_{k=0}^{\N}$ via real-noisy observations $\leftl \y_k \rightl_{k=0}^{\N}$ \cite{kalnay2003atmospheric,lahoz2010data,reichle2008data}, where $\x^b_k \in \Re^{\Nstate \times 1}$ is well-known as the forecast state, $k$ stands for time index  where observations are collected, for $0 \le k \le \N$, $\Nstate$ denotes the state size, the size of the assimilation window reads $\N$, $\y_k \in \Re^{\Nobs \times 1}$ is a vector of observations, and $\Nobs$ is the number of observations. Forecasts are typically obtained by numerical models which minimizes the behaviour of the ocean and/or the atmosphere: 
\begin{eqnarray}
\label{eq:numerical-model}
\x_{k} = \Mo_{t_{k-1} \rightarrow t_k} \leftp \x_{k-1} \rightp \,\, \text{ for $0 \le k-1 \le \N$, and $\x_k \in \Re^{\Nstate \times 1}$} \,,
\end{eqnarray}
where $\Mo:\Re^{\Nstate \times 1} \rightarrow \Re^{\Nstate \times 1}$ is the non-linear model operator (typically, given in the form of partial and ordinary differential equations). In the context of Four-Dimensional Variational (4D-Var) data assimilation methods \cite{huang2009four},cost functions of the form \cite{tr2006accounting,tremolet2007incremental}:
\begin{eqnarray}
\label{eq:4D-Var-function}
\displaystyle
\J(\x_0) = \frac{1}{2} \cdot \norm{\x_0-\x_0^b}^2_{\B^{-1}_0} +\frac{1}{2} \cdot \sum_{k=0}^{\N-1} \norm{\y_k-\Ho_k \leftp \x_k \rightp}^2_{\R_k^{-1}} \,,
\end{eqnarray}
are considered to estimate the initial state $\x_0 \in \Re^{\Nstate \times 1}$ which best fit a given assimilation window $ \leftl \y_k \rightl_{k=0}^{\N}$. The optimization problem to solve reads:
\begin{eqnarray}
\label{eq:4D-Var-Opt-Problem}
\displaystyle
\x^a_0 = \arg \underset{\x_0}{\min} \, \J \leftp \x_0 \rightp \,, \text{subject to $\x_{k} = \Mo_{t_{k-1} \rightarrow t_k} \leftp \x_{k-1}\rightp$} \,.
\end{eqnarray}
A potential drawback of 4D-Var based methods is the use of adjoint models. In practice, these models can be labor-intensive to develop and computationally expensive to run \cite{ruiz2016derivative}. Thus, ensemble-based methods can be employed under 4D-Var operational scenarios to avoid the use of adjoint models \cite{harlim2007four,miyoshi2012local}. The main idea behind these approaches is to build ensemble sub-spaces onto which analysis increments can be estimated \cite{wang2007comparison}. For non-linear observation operators, the underlying error distribution of background errors can be non-Gaussian and as a direct consequence, no closed-form expression are available for computing posterior modes in \eqref{eq:4D-Var-Opt-Problem} \cite{arulampalam2002,tr2006accounting}. Hence, iterative methods such as Trust Region \cite{more1983computing} and Line Search \cite{more1994line} can be employed to estimate posterior modes of the implicit error distribution in \eqref{eq:4D-Var-function}. We think that there is an opportunity to derive a line-search-based optimization method for the solution of \eqref{eq:4D-Var-Opt-Problem} via control-space approximations (i.e., based on ensembles). In this manner, an ensemble of model realizations can be employed to capture the error dynamics across all observation times and, even more, to build spaces onto which analysis increments can be estimated. Simple ensemble anomalies can span Control-spaces or robust square-root approximations of background-error-covariance matrices via a modified Cholesky decomposition  \cite{bickel2008covariance,rothman2009generalized}. In the last estimator, its special structure can allow matrix-free 4D-EnKF implementations.

This paper's outline is as follows: in Section \ref{sec:preliminaries}, theoretical and practical aspects of variational and ensemble data assimilation methods are briefly discussed in the linear and the non-linear cases. In Section \ref{sec:proposed-method}, two efficient implementations of 4D-EnKF formulations are proposed via line-search optimization, the first one makes use of a modified Cholesky decomposition to build control spaces while the last one uses an ensemble of anomalies to estimate posterior increments. Their computational costs are detailed as well as the convergence of the line-search based optimization method. In Section \ref{sec:experimental-settings}, experimental tests are performed to assess the accuracy of the proposed filter implementations, the Lorenz 96 model is employed as our surrogate model for the experiments. Conclusions of this research are stated in Section \ref{sec:conclusions}.

\section{Preliminaries}
\label{sec:preliminaries}

This Section discusses some related topics for the development of our variational filter.

\section{The Ensemble Kalman Filter}

The Ensemble Kalman Filter (EnKF) \cite{lorenc2003potential,evensen2003ensemble} is a sequential ensemble based method which employes an ensemble of model realizations to estimate moments of prior error distributions \cite{houtekamer1998data,stroud2018bayesian}:
\begin{eqnarray}
\label{eq:background-ensemble}
\X_k^b = \leftb \x^{b[1]}_k,\, \x^{b[2]}_k,\, \ldots,\, \x^{b[\Nens]}_k \rightb \in \Re^{\Nstate \times \Nens}\,,
\end{eqnarray}
where $\x^{b[e]}_k \in \Re^{\Nstate \times 1}$ denotes the $e$-th ensemble member, for $1 \le e \le \Nens$, at time $k$, for $0 \le k \le \N$. The empirical moments of \eqref{eq:background-ensemble} are employed to estimate the forecast state $\x^b_k$:
\begin{eqnarray}
\label{eq:ensemble-mean}
\displaystyle
\xm^b_k = \frac{1}{\Nens} \cdot \sum_{e=1}^{\Nens} \x^{b[e]}_k  \in \Re^{\Nstate \times 1} \,,
\end{eqnarray}
and the backgrund error covariance matrix $\B$:
\begin{eqnarray}
\label{eq:ensemble-covariance}
\displaystyle
\P^b_k = \frac{1}{\Nens-1} \cdot \DX^b_k \cdot \leftb \DX^b_k \rightb^T \in \Re^{\Nstate \times \Nstate} \,,
\end{eqnarray}
where the matrix of member deviations reads:
\begin{eqnarray}
\label{eq:deviations}
\displaystyle
\DX^b_k = \X^b_k-\xm^b_k \cdot \ones^T \in \Re^{\Nstate \times \Nens} \,.
\end{eqnarray}
The posterior ensemble can then be built by using synthetic observations:
\begin{eqnarray}
\label{eq:analysis-ensemble-k}
\X^a_k = \X^b_k + \DX^a_k \,,
\end{eqnarray}
where the analysis updates can be obtained by solving the following linear system:
\begin{eqnarray}
\leftb \leftb \P^b_k \rightb^{-1} + \H_k^T \cdot \R_k^{-1} \cdot \H_k \rightb \cdot \DX^a_k = \H_k^T \cdot \R_k^{-1} \cdot \D^s_k \in \Re^{\Nstate \times \Nens} \,,
\end{eqnarray}
the innovation matrix reads $\D_k^s \in \Re^{\Nobs \times \Nens}$ whose $e$-th column $\y_k-\H_k \cdot \x^{b[e]}_k + \errobs_k^{[e]} \in \Re^{\Nobs \times 1}$ is a synthetic observation with $\errobs_k^{[e]} \sim \Nor \leftp \zero_{\Nobs},\, \R_k \rightp$. In operational DA scenarios, the ensemble size can be lesser than model dimensions by order of magnitudes and as a direct consequence, sampling errors impact the quality of analysis increments. To counteract the effects of sampling noise, localizations methods are commonly employed \cite{greybush2011balance,chen2010cross}. In the context of covariance tappering, the use of the \textit{spatial-predecessors} concept can be employed to obtain sparse estimators of precision matrices \cite{levina2008sparse}. The predecessors of model component $i$, from now on $P(i,\,\ra)$, for $1 \le i \le \Nstate$ and a radius of influence $\ra \in \Zr^{+}$, are given by the set of components whose labels are lesser than that of the $i$-th one. 

In the EnKF based on a modified Cholesky decomposition (EnKF-MC) \cite{nino2017parallel,nino2018ensemble} the following estimator is employed to approximate the precision covariance matrix of the background error distribution \cite{bickel2008regularized}:
\begin{eqnarray}
\label{eq:mc-ednr}
\BE^{-1}_k = \LE_k^T \cdot \DE^{-1}_k \cdot \LE_k \in \Re^{\Nstate \times \Nstate} \,,
\end{eqnarray}
where the Cholesky factor $\LM_k \in \Re^{\Nstate \times \Nstate}$ is a lower triangular matrix,
\begin{eqnarray}
\label{eq:L-factor}
\leftl \LE_k \rightl_{i,v} = \begin{cases}
-\beta_{i,v,k} & \,,\, v \in P(i,\ra) \\
1 & \,,\, i=v \\
0 &\,,\, otherwise
\end{cases} \,,
\end{eqnarray}
whose non-zero sub-diagonal elements $\beta_{i,v,k}$ are obtained by fitting models of the form,
\begin{eqnarray}
\label{eq:fitting-models-of-the-form}
\displaystyle
{\x_{[i]}^T}_k = \sum_{v \in P(i,\,\ra)} \beta_{i,v,k} \cdot {\x_{[v]}^T}_k + {\erra_i}_k \in \Re^{\Nens \times 1} \,, \text{ $1 \le i \le \Nstate$}\,,
\end{eqnarray}
where ${\x_{[i]}^T}_k \in \Re^{\Nens \times 1}$ denotes the $i$-th row (model component) of the ensemble \eqref{eq:background-ensemble}, components of vector ${\erra_i}_k \in \Re^{\Nens \times 1}$ are samples from a zero-mean Normal distribution with unknown variance $\sigma^2_k$, and $\D_k \in \Re^{\Nstate \times \Nstate}$ is a diagonal matrix whose diagonal elements read,
\begin{eqnarray}
\label{eq:diagonal-elements-of-D}
\leftl \D_k \rightl_{i,i} &=& \evar \leftp {\x_{[i]}^T}_k -\sum_{v \in P(i,\,\ra)} \beta_{i,v,k} \cdot {\x_{[j]}^T}_k  \rightp^{-1} \\
& \approx & \var \leftp {\erra_i}_k \rightp^{-1} = \frac{1}{\sigma_k^2} >0 \,, \text{ with $\leftl \D_k \rightl_{1,1} = \evar \leftp {\x_{[1]}^T}_k \rightp^{-1}$}\,,
\end{eqnarray}
where $\var(\bullet)$ and $\evar(\bullet)$ denote the actual and the empirical variances, respectively. The analysis equations can then be written as discussed in \cite{nino2018ensemble}.

\subsection{Maximum Likelihood Ensemble Filter}
\label{subsec:MLEF}

The filter derivation in the EnKF context is based on the assumptions that prior and observational errors are Gaussian. However, for non-linear observation operators (i.e., information brought by satellite radiances), observational errors are not Gaussian anymore and as a direct consequence, EnKF based methods can fail to obtain reasonable estimates of posterior error distributions. The Maximum Likelihood Ensemble Filter (MLEF) \cite{zupanski2005} is a numerical method that minimizes the Three-Dimensional Variational (3D-Var) cost function:
\begin{eqnarray}
\label{eq:3D-Var-function}
\displaystyle
\J(\x_k) = \frac{1}{2} \cdot \norm{\x_k-\xm_k^b}^2_{\B^{-1}_k} +\frac{1}{2} \cdot \norm{\y_k-\Ho_k \leftp \x_k \rightp}^2_{\R_k^{-1}} \,,
\end{eqnarray}
onto the ensemble space. For this purpose, analysis increments are constrained to the space spanned by the ensemble of anomalies \eqref{eq:deviations}, this is:
\begin{eqnarray}
\label{eq:3D-Var-ensemble-space}
\displaystyle
\x_k = \xm^b_k + \DX_k \cdot \vw \,,
\end{eqnarray}
which results in a cost function of the form:
\begin{eqnarray}
 \displaystyle \nonumber
\J(\xm^b_k + \DX_k \cdot \vw) &=& \frac{\Nens-1}{2} \cdot \norm{\vw}^2\\ \label{eq:3D-Var-function-ensemble-space}
&+&\frac{1}{2} \cdot \norm{\y_k-\Ho_k \leftp \xm^b_k + \DX_k \cdot \vw  \rightp}^2_{\R_k^{-1}} \,,
\end{eqnarray}
and corresponding optimization problem:
\begin{eqnarray}
\label{eq:3D-Var-optimization}
\displaystyle
\vw^{*} = \arg \underset{\vw}{\min} \, \J(\xm^b_k + \DX_k \cdot \vw)\,.
\end{eqnarray}
Note that, since all computations are performed onto the ensemble space, the computational effort of the MLEF becomes:
\begin{eqnarray}
\label{eq:MLEF-computational-cost}
\displaystyle
\BO{\Nens^{3} + \Nstate \cdot \Nobs \cdot \Nens} \,,
\end{eqnarray}
once all analysis increments are mapped back onto the model space (and the posterior ensemble is built about the estimated posterior mode), which is linear with regard to the model dimension.

\subsection{Four Dimensional Variational Ensemble Kalman Filter }

In Four Dimensional Variational Ensemble Kalman Filter (4D-EnKF) methods \cite{liu2008ensemble,lorenc2015comparison}, the model trajectory in \eqref{eq:4D-Var-Opt-Problem}  is constrained to the space spanned by the background ensemble members, this is:
\begin{eqnarray}
\label{eq:ensemble-space}
\displaystyle
\x_k = \xm^b_k + \DX_k \cdot \vw \,,
\end{eqnarray}
where $\vw \in \Re^{\Nens \times 1}$ is a vector in redundant coordinates to be determined later. Making the assumption that all sources of errors are Gaussian (i.e., the observation operator is linear), and by replacing \eqref{eq:ensemble-space} into the equation \eqref{eq:4D-Var-function} one obtains:
\begin{eqnarray}
\nonumber
\J(\x_0) &=& \J \leftp \xm^b_0 + \DX_0\cdot \vw \rightp \\ \label{eq:4D-Var-ensemble-space}
 &=& \JE(\vw) = \frac{(\Nens-1)}{2} \cdot \norm{\vw}^2 + \frac{1}{2} \cdot \sum_{k=0}^{\N} \norm{\dv_k - \Q_k \cdot \vw}^2_{\R_k^{-1}} \,,
\end{eqnarray}
where $\dv_k = \y_k - \H_k \cdot \xm^b_k \in \Re^{\Nobs \times 1}$ is the innovation vector and $\Q_k = \H_k \cdot \DX_k^b \in \Re^{\Nobs \times \Nens}$. The optimal values of the control variable $\vw$ is then sought in order to estimate the initial analysis state:
\begin{eqnarray}
\label{eq:new-cost-function}
\displaystyle
\vw^{*} = \arg\,\underset{\vw}{\min} \, \JE(\vw) \,.
\end{eqnarray}
The gradient of \eqref{eq:4D-Var-ensemble-space} equals:
\begin{eqnarray}
\displaystyle \nonumber
\nabla_{\vw} \JE(\vw) &=& (\Nens-1) \cdot \vw - \sum_{k=0}^{\N} \Q_k^T \cdot \R_k^{-1} \cdot \leftb \dv_k - \Q_k \cdot \vw \rightb \\ \nonumber
&=& \leftb (\Nens-1) \cdot \I + \sum_{k=0}^{\N} \Q_k^T \cdot \R_k^{-1} \cdot \Q_k \rightb \cdot \vw \\ \label{eq:gradient-4D-Var-ensemble-space}
&-& \sum_{k=0}^{\N} \Q_k^T \cdot \R_k^{-1} \cdot \dv_k \in \Re^{\Nens \times 1} \,,
\end{eqnarray}
and by setting this gradient to zero, the optimal weights \eqref{eq:new-cost-function} reads:
\begin{eqnarray}
\label{eq:optimal-weights}
\displaystyle
\vw^{*} = \leftb (\Nens-1) \cdot \I + \sum_{k=0}^{\N} \Q_k^T \cdot \R_k^{-1} \cdot \Q_k \rightb^{-1} \cdot \sum_{k=0}^{\N} \Q_k^T \cdot \R_k^{-1} \cdot \dv_k \in \Re^{\Nens \times 1}  \,,
\end{eqnarray}
from which the initial analysis state can be estimated:
\begin{eqnarray}
\label{eq:initial-analysis-estimation}
\displaystyle
\xm_0^a = \xm_0^b + \DX_0^b \cdot \vw^{*} \,.
\end{eqnarray}
Similar to the MLEF method, the computational effort of 4D-EnKF based methods are linearly bounded by the model dimension; it can be easily seen that all computations are performed onto the ensemble space. Note that, in this formulation, posterior estimates can be highly biased regarding the actual state when the observation operator is non-linear. 

\subsection{Numerical Optimization}

In numerical optimization \cite{vanderplaats1984numerical,wright1999numerical}, commonly, problems of the form:
\begin{eqnarray*}
\x = \arg\,\underset{\x}{\min} \, f(\x) \,,
\end{eqnarray*}
where $f(\x):\Re^{n \times 1} \rightarrow \Re$ is a smooth function, are solved by iterative schemes such as
\begin{eqnarray}
\x^{(u+1)} = \x^{(u)} + \dx^{(u)} \,,
\end{eqnarray}
wherein $\dx^{(u)}$ is a search direction, for instance, the steepest descent direction \cite{savard1994steepest,hager2006survey,fletcher1964function,lewis2000direct}
%
the Newton's step \cite{battiti1992first,grippo1989truncated,pan1999newton},
%
or quasi-Newton based method \cite{shanno1970conditioning,nocedal1980updating,loke1996rapid}.
%
A concise survey of Newton based methods can be consulted in \cite{knoll2004jacobian}. Since step sizes in steepest descent directions can be too large, line search \cite{grippo1986nonmonotone,uschmajew2014line,hosseini2018line} or trust region \cite{conn2000trust,more1983computing,curtis2017trust} methods can be successfully used in order to optimally choose the step size and therefore, to ensure global convergence of iterates to stationary points. This holds as long as some assumptions over functions, gradients, and (potentially) Hessians are preserved \cite{shi2004convergence}. In the context of line search, the following assumptions are commonly done:
\begin{enumerate}
\item \label{cond:cond1} The function $f(\x)$ has lower bound on $\Omega_0 = \{\x \in \Re^{n \times 1},\, f(\x) \le f(\x_0)\}$, where $\x_0 \in \Re^{n \times 1}$ is available. 
\item \label{cond:cond2} The gradient $\nabla f(\x)$ is Lipschitz continuous in an open convex set $B$ which contains $\Omega_0$,
\begin{eqnarray*}
\norm{\nabla f(\x)-\nabla f(\z)} \le L \cdot \norm{\x-\z},\, \text{ for $\x,\,\z \in B$,  and $L > 0$}.
\end{eqnarray*}
These together with iterates of the form,
\begin{eqnarray}
\label{eq:iterates-line-search}
\displaystyle
\x^{(u+1)} = \x^{(u)} + \alpha \cdot \dx^{(u)} \,,
\end{eqnarray}
ensure global convergence \cite{zhou2017matrix} as long as $\alpha$ is chosen as an (approximated) minimizer of
\begin{eqnarray}
\label{eq:line-search-cost-function}
\displaystyle
\alpha^{*} = \arg\,\underset{\alpha \ge 0}{\min} \, f \leftp  \x^{(u)} + \alpha \cdot \dx^{(u)}\rightp \,.
\end{eqnarray}
\end{enumerate}

\section{Proposed Method}
\label{sec:proposed-method}

In this section, we develop an efficient and practical implementation of a 4D-EnKF method via a modified Cholesky decomposition for non-linear observation operator.

\subsection{Filter Derivation}

Our main goal is to find an initial analysis ensemble $\X^a_0 \in \Re^{\Nstate \times \Nens}$ whose model trajectories account for the information brought by observations within a given assimilation window. To accomplish this, we proceed as follows: given an assimilation window of $\N$ observations (with no loss of generality, evenly distributed in time), by using the numerical model \eqref{eq:numerical-model}, the initial background ensemble $\X_0^b$ is propagated in time and snapshots of its time evolution are taken at observation times. Denote  by $k$, the $k$-th step where an observation is available, for $0 \le k \le \N$, at each observation time, the background ensemble \eqref{eq:background-ensemble} can be employed to estimate a full-rank square-root approximation of the precision covariance matrix $\B_k^{-1}$ via a modified Cholesky decomposition \eqref{eq:mc-ednr}:
\begin{eqnarray}
\label{eq:square-root-app}
\displaystyle
\BE^{-1/2}_k = \LE^T_k \cdot \DE^{-1/2}_k \in \Re^{\Nstate \times \Nstate} \,.
\end{eqnarray}
This square-root estimation can then be employed to build a control space onto which analysis increments can be estimated:
\begin{eqnarray*}
\x_k-\xm^b_k \in \range \leftl \BE^{1/2}_k \rightl \,,
\end{eqnarray*}
or equivalently:
\begin{eqnarray}
\label{eq:redundant-coordinates}
\displaystyle
\x_k = \xm^b_k + \BE^{1/2}_k \cdot \va \in \Re^{\Nstate \times 1} \,,
\end{eqnarray}
where $\va \in \Re^{\Nstate \times 1}$ is a vector in redundant coordinates to be determined later. Since the background error covariance matrix onto the control space \eqref{eq:redundant-coordinates} is nothing but the identity matrix, the following error statistics hold for the prior weights $\va^b$:
\begin{eqnarray*}
\displaystyle
\va^b \sim \Nor \leftp \zero,\, \I \rightp \,.
\end{eqnarray*}
Due to this, the 4D-Var cost function \eqref{eq:4D-Var-function} onto the space \eqref{eq:redundant-coordinates} can be written as follows:
\begin{eqnarray} \nonumber
\J \leftp \x_0 \rightp &=& \J \leftp \xm^b_0 + \BE^{1/2}_0 \cdot \va \rightp = \frac{1}{2} \cdot \norm{\va}^2 \\ \label{eq:4D-Var-MC}
&+& \frac{1}{2} \cdot \sum_{k=0}^\N \norm{\y_k - \Ho_k \leftp \xm^b_k + \BE^{1/2}_k \cdot \va \rightp }_{\R_k^{-1}}^2 \,,
\end{eqnarray}
The adjoint-free optimization problem to solve reads:
\begin{eqnarray}
\label{eq:optimization-problem-MC}
\displaystyle
\va^{a} = \arg \, \underset{\va}{\min} \, \J \leftp \xm^b_0 + \BE^{1/2}_0 \cdot \va \rightp\,.
\end{eqnarray}
The next section develops a Line-Search method for solving the optimization problem \eqref{eq:optimization-problem-MC}.

\subsection{An Adjoint-Free 4D-Var Method Via Modified Cholesky Decomposition (4D-Var-MC)}
\label{subsec:4dvarmc}
%

To approximate a solution of \eqref{eq:optimization-problem-MC}, we propose iterates of the form:
\begin{eqnarray}
\label{eq:sub-space-ensemble-iterative}
\displaystyle
\x^{(u+1)}_k = \x^{(u)}_k + \BE^{1/2}_k \cdot \bal_u\,, \text{ for $0 \le u \le U$}\,,
\end{eqnarray}
where $u$ denotes iteration step, $U$ stands for the number of iterations, and $\bal_u \in \Re^{\Nstate \times 1}$ is a vector in redundant coordinates to be determinated later. The background trajectory is set as the initial solution $\x^{(0)}_k = \xm^b_k$, for $0 \le k \le \N-1$, and therefore $\va_0 = \zero$. Note that, at any iteration $u$, 
\begin{eqnarray}
\displaystyle \nonumber
\x_k^{(u+1)} &=& \x^{(0)}_k+  \sum_{v=0}^{u} \BE^{1/2}_k \cdot  \bal_v = \xm^b_k + \sum_{v=0}^{u-1} \BE^{1/2}_k \cdot  \bal_v + \BE^{1/2}_k \cdot  \bal_u\, \\ \label{eq:to-replace}
&=&  \xm^b_k + \BE^{1/2}_k  \cdot \leftb \sum_{v=0}^{u-1}  \bal_v +\bal_u \rightb =  \xm^b_k + \BE^{1/2}_k  \cdot \leftb \bv_u +\bal_u \rightb  \,
\end{eqnarray}
where we have let $\bv_u = \sum_{v=0}^{u-1} \bal_v$. By replacing \eqref{eq:sub-space-ensemble-iterative} and \eqref{eq:to-replace} in \eqref{eq:4D-Var-function} one obtains:
\begin{eqnarray*}
\label{eq:4D-alpha}
\displaystyle
\JT \leftp \va_u\rightp = \J \leftp \x^{(u)}_k + \BE^{1/2}_k \cdot \bal_u \rightp &=& \frac{1}{2} \cdot \norm{\bv_u+\bal_u}^2 \\
 &+& \frac{1}{2} \cdot \sum_{k=0}^{\N-1} \norm{\y_k-\Ho_k \leftp \x^{(u)}_k + \BE^{1/2}_k \cdot \bal_u\rightp}^2_{\R_k^{-1}} \,.
\end{eqnarray*}
At iteration $u$, we choose $\va_u$ in \eqref{eq:sub-space-ensemble-iterative} as the value that solves the following optimization problem:
\begin{eqnarray}
\label{eq:problem-to-solve-u}
\displaystyle
\va_u^{*} = \arg\,\underset{\va_u}{\min} \, \J \leftp \x^{(u)}_k + \BE^{1/2}_k \cdot \bal_u \rightp \,.
\end{eqnarray}
A solution of this problem can be approximated as follows: consider the linearization of $\Ho_k$ about $\x^{(u)}_k$, at iteration $u$ 
\begin{eqnarray*}
\label{eq:linearization}
\displaystyle
\Ho_k \leftp \x_k \rightp  \approx \HG^{(u)}_k \leftp\x_k \rightp &=& \Ho_k \leftp \x^{(u)}_k  \rightp + \H^{(u)}_k \cdot \leftb \x_k-\x^{(u)}_k  \rightb\\
&=& \Ho_k \leftp \x^{(u)}_k  \rightp + \H^{(u)}_k \cdot \BE^{1/2}_k\cdot\bal_u \,,
\end{eqnarray*}
where $\H^{(u)}_k  \in \Re^{\Nobs \times \Nstate}$ is the Jacobian of $\Ho_k(\x_k)$ at $\x^{(u)}_k$, 
then, at iteration $u$, we make use of the next quadratic approximation:
\begin{eqnarray*}
\label{eq:4D-alpha-linear}
\displaystyle
\JT \leftp \va_u\rightp & \approx & \JN(\va_u) \\
&=& \frac{1}{2} \cdot \norm{\bv_u+\bal_u}^2 + \frac{1}{2} \cdot \sum_{k=0}^{\N-1} \norm{\y_k-\HG^{(u)}_k \leftp \x^{(u)}_k + \BE^{1/2}_k \cdot \bal_u \rightp}^2_{\R_k^{-1}} \\
&=& \frac{1}{2} \cdot \norm{\bv_u+\bal_u}^2 + \frac{1}{2} \cdot \sum_{k=0}^{\N-1} \norm{\dv^{(u)}_k-\Q_k^{(u)} \cdot \bal_u}^2_{\R_k^{-1}} 
\end{eqnarray*}
where $\dv^{(u)}_k = \y_k-\Ho_k \leftp \x_k^{(u)} \rightp \in \Re^{\Nobs \times 1}$, and $ \Q_k^{(u)}= \H^{(u)}_k \cdot \BE^{1/2}_k \in \Re^{\Nobs \times \Nstate}$. Thus,
\begin{eqnarray*}
\label{eq:JN}
\displaystyle
\nabla \JN_{\bal_u}(\va_u) &=& \bv_u+\bal_u- \sum_{k=0}^{\N} \leftb \Q_k^{(u)}  \rightb^T \cdot \R^{-1}_k \cdot \leftb \dv^{(u)}_k- \Q_k^{(u)}  \cdot \bal_u\rightb \\
&=& \leftb \I + \sum_{k=0}^{\N} \leftb \Q_k^{(u)}  \rightb^T \cdot \R^{-1}_k \cdot  \Q_k^{(u)}\rightb \cdot \bal_u\\
&-& \leftb -\bv_u+\sum_{k=0}^{\N}\leftb \Q_k^{(u)}  \rightb^T  \cdot \R^{-1}_k \cdot \dv^{(u)}_k\rightb \in \Re^{\Nens \times 1} \,,
\end{eqnarray*}
and 
\begin{eqnarray*}
\displaystyle
\nabla \JN_{\bal_u,\bal_u}(\bal_u) &=&  \I + \sum_{k=0}^{\N} \leftb \Q_k^{(u)}  \rightb^T \cdot \R^{-1}_k \cdot \Q_k^{(u)} \,.
\end{eqnarray*}
We consider the search direction as the solution of $\nabla \JN_{\bal_u}(\hat{\bal}_u )={\boldsymbol 0}$, for computing partial analysis increments onto the control space, this is 
\begin{eqnarray*}
 \hat{\bal}_u = \leftb  \nabla \JN_{\bal_u,\bal_u}(\hat{\bal}_u) \rightb^{-1} \cdot \leftb -\bv_u+\sum_{k=0}^{\N}\leftb \Q_k^{(u)}  \rightb^T  \cdot \R^{-1}_k \cdot \dv^{(u)}_k\rightb \,, 
\end{eqnarray*}
but, since we employ a linear approximation of the observation operator to approximate a descent direction of \eqref{eq:4D-Var-Opt-Problem} , we perform a line search to estimate the partial analysis increments:
\begin{eqnarray}
\label{eq:line-search}
\displaystyle
\rho^{*}_u = \arg\,\underset{\rho_u \in [0,\,1]}{\min}\, \J \leftp \x^{(u)}_k + \BE^{1/2}_k \cdot \leftb \rho_u \cdot \hat{\bal}_u \rightb  \rightp \,,
\end{eqnarray}
and therefore, the next iterates reads:
\begin{eqnarray}
\x^{(u+1)}_k = \x^{(u)}_k + \B^{1/2}_k \cdot \leftb \rho^{*}_u \cdot  \hat{\bal}_u \rightb
\end{eqnarray}
where $\va_u^{*} \approx \rho^{*}_u \cdot  \hat{\bal}_u$, this yields to $\bv_u^{*} \approx \sum_{v=1}^{u-1} \,\rho^{*}_v \cdot  \hat{\bal}_v $. This process is repeated until some stopping criteria is satisfied (i.e., a maximum number of iterations $U$ is reached). The resulting analysis state reads:
\begin{eqnarray*}
\displaystyle
\xm_0^a = \xm_0^b + \B^{1/2}_0 \cdot \sum_{u=1}^{U}\, \rho^{*}_u \cdot  \hat{\bal}_u \,,
\end{eqnarray*}
which is nothing but a linear combination of the approximated search directions among iterations. Thus, we can approximate $\va^a$ in \eqref{eq:optimization-problem-MC} as follows:
\begin{eqnarray*}
\displaystyle
\va^a \approx \sum_{u=1}^{U} \,\rho^{*}_u \cdot \hat{\bal}_u \,.
\end{eqnarray*}

 We can employ any line-search rule to solve the optimization problem \eqref{eq:line-search}. Putting it all together, the optimization steps are detailed in the Algorithm \ref{alg:4dvarmcd}. The posterior ensemble onto the control space can then be built by using a square root approximation of the information matrix, its can be easily shown that the posterior error statistics read:
\begin{eqnarray}
\label{eq:posterior-weights}
\displaystyle
\va^{a[e]} \sim \Nor \leftp \va^{a},\, \leftb \I + \sum_{k=0}^{\N} \leftb  \Q_k^{(U)} \rightb^T \cdot \R_k^{-1} \cdot \Q^{(U)}_k   \rightb^{-1} \rightp \,, \text{ for $1 \le e \le \Nens$} \,,
\end{eqnarray}
with corresponding analysis members in the model space:
\begin{eqnarray*}
\x^{a[e]}_0 = \xm^{a}_0 + \BE_0^{1/2} \cdot \va^{a[e]} \,.
\end{eqnarray*}
The initial ensemble members are then propagated in time from which the optimal (ensemble) trajectory is recovered.
\begin{algorithm}
\begin{algorithmic}[1]
\Function{find\_initial\_analysis\_ensemble}{$\left\{\y_k,\xm_k^b,\BE^{1/2}_k,\R^{-1}_k\right\}^\M_{k=0}$} 
\State $\va^a \gets {\boldsymbol 0}$
\For{$k \gets 1 \to M$}
\State $\x^{(0)}_k \gets \xm_k^b$
\EndFor
\For{$u \gets 0 \to U$}
\State $\bv^{*}_u \gets \va^a$
\For{$k \gets 1 \to \N$}
\State $\dv^{(u)}_k \gets \y_k-\Ho_k \leftp \x_k^{(u)}\rightp$
\State $\Q^{(u)}_k \gets \H^{(u)}_k \cdot \BE^{1/2}_k$ \Comment{$\H^{(u)}_k$ is the Jacobian of $\Ho_k(\x_k)$ at $\x^{(u)}_k$}
\EndFor
\State $\JN_{\bal_u,\bal_u}(\hat{\bal}_u) \gets \I + \sum_{k=0}^{\N} \leftb \Q_k^{(u)}  \rightb^T \cdot \R^{-1}_k \cdot \Q_k^{(u)}$
\State $\hat{\bal}_u \gets \leftb  \nabla \JN_{\bal_u,\bal_u}(\hat{\bal}_u) \rightb^{-1} \cdot \leftb -\bv_u+\sum_{k=0}^{\N}\leftb \Q_k^{(u)}  \rightb^T  \cdot \R^{-1}_k \cdot \dv^{(u)}_k\rightb$
\State $\rho^{*}_u \gets \arg\,\underset{\rho_u \in [0,\,1]}{\min}\, \J \leftp \x^{(u)}_k + \BE^{1/2}_k \cdot \leftb \rho_u \cdot \hat{\bal}_u \rightb  \rightp$
\For{$k \gets 1 \to \N$}
\State $\x^{(u+1)}_k \gets \x^{(u)}_k + \B^{1/2}_k \cdot \leftb \rho^{*}_u \cdot  \hat{\bal}_u \rightb$
\EndFor
\State $\va_u^{*} \gets \rho^{*}_u \hat{\bal}_u$
\State $\va^a \gets \va^a+\va_u^{*}$
\EndFor
\State $\xm_0^a \gets \x^{(U)}_0$
\For{$e \gets 1 \to \Nens$} 
\State $\va^{a[e]} \sim \Nor \leftp \va^{a},\, \leftb \I + \sum_{k=0}^{\N} \leftb  \Q_k^{(U)} \rightb^T \cdot \R_k^{-1} \cdot \Q^{(U)}_k   \rightb^{-1} \rightp$
\State $\x^{a[e]}_0 \gets \xm^{a}_0 + \BE_0^{1/2} \cdot \va^{a[e]}$
\EndFor

\State \Return $\{\x^{a[e]}_0\}^\Nens_{e=1}$
\EndFunction
\end{algorithmic}
\caption{An Adjoint-Free 4D-Var Method Via Modified Cholesky Decomposition (4D-Var-MC)}
\label{alg:4dvarmcd}
\end{algorithm}
Matrix inversions in eht 4D-Var-MC can be avoided by considering at each optimization step the matrix-free 4D-Var formulation \cite{nino2020adjoint}.

\subsection{An Adjoint-Free 4D-Var Method Via MLEF (4D-Var-MLEF)} 

The proposed 4D-Var-MC can be similarly derived onto the ensemble space by employing the matrix of anomalies \eqref{eq:deviations} as our set of basis vectors at observation times. We then can propose iterates of the form:
\begin{eqnarray}
\label{eq:sub-space-mlef-iterative}
\displaystyle
\x^{(u+1)}_k = \x^{(u)}_k + \DX_k \cdot \vgam_u\,, \text{ for $0 \le u \le U$}\,,
\end{eqnarray}
where $\vgam_u \in \Re^{\Nens \times 1}$. We perform analogous derivations as in \eqref{eq:to-replace}--\eqref{eq:4D-alpha-linear}, replacing the control space $\BE^{1/2}_k$ by the ensemble space $\DX_k$. Search directions are then approximated by considering linearizations of the form:
\begin{eqnarray*}
\label{eq:linearization-mlef}
\displaystyle
\Ho_k \leftp \x_k \rightp  \approx \HG^{(u)}_k \leftp\x_k \rightp &=& \Ho_k \leftp \x^{(u)}_k  \rightp + \H^{(u)}_k \cdot \DX_k\cdot\bal_u \,,
\end{eqnarray*}
and similar to \eqref{eq:line-search}, a line-search method can be employed to control step lengths among iterations. This algorithm becomes an adjoint-free method based on the MLEF method discussed in section \ref{subsec:MLEF}; we name it the 4D-Var-MLEF. Note that, since all computations are performed onto the ensemble space, the computational cost of estimating analysis increments together with analysis updates reads:
\begin{eqnarray*}
\displaystyle
\BO{\Nens^3 + \Nens^2 \cdot \Nobs + \Nens \cdot \Nobs \cdot \Nstate } \,.
\end{eqnarray*}
As a drawback, the 4D-Var-MLEF suffers from the MLEF limitations regarding global convergence. 

\subsection{Global Convergence of the 4D-Var-MC}

For proving the convergence of the line-search method during the assimilation of observations in the 4D-Var-MC, we consider the model \eqref{cond:cond1}, \eqref{cond:cond2}, and we assume 
\begin{eqnarray}
\label{eq:theorem-ortho}
\displaystyle
\nabla \J \leftp \x^{(u)}\rightp^T  \cdot \leftb \BE^{1/2}_0 \cdot \va^{*}_u \rightb < 0,\, \text{ for $1 \le k \le K$} \,.
\end{eqnarray}
The next Theorem states the necessary conditions to ensure global convergence for the line search method during the assimilation of observations in the 4D-Var-MC method.

\begin{theorem}
\label{theo:TS-MGA}
If \eqref{cond:cond1}, \eqref{cond:cond2}, and \eqref{eq:theorem-ortho} hold, 
then the line search  in the 4D-Var-MC generates an infinite sequence $\leftl \x^{(u)}\rightl_{u=0}^{\infty}$, then 
\begin{eqnarray}
\label{eq:to-prove}
\displaystyle
\underset{u \rightarrow \infty}{\lim} \leftb \frac{ - \nabla \J \leftp \x^{(u)} \rightp^T \cdot  \BE^{1/2}_0 \cdot \va^{*}_u }{ \norm{\BE^{1/2}_0 \cdot \va^{*}_u}} \rightb^2 = 0
\end{eqnarray}
holds.
\end{theorem}

\begin{proof}
By Taylor series and the Mean Value Theorem we know that,
\begin{eqnarray*}
\J \leftp \x^{(u)}_0 + \rho_u^{*} \cdot  \B^{1/2}_0 \cdot \va^{*}_u \rightp &=& \J \leftp \x^{(u)}_0 \rightp \\
&+& \rho^{*} \cdot \int_{0}^1 \, \nabla \J \leftp \x^{(u)}_0 + \rho^{*}_u \cdot t \cdot \B^{1/2}_0 \cdot \va^{*}_u \rightp^T \\
& \cdot & \B^{1/2}_0 \cdot \va^{*}_u \cdot dt \,,
\end{eqnarray*}
where $\va^{*}_u$ is given by \eqref{eq:problem-to-solve-u}, and therefore,
\begin{eqnarray*}
\J \leftp \x^{(u)}_0  \rightp - \J \leftp \x^{(u+1)}_0  \rightp &\ge & - \rho^{*}_u \cdot \int_{0}^1 \, \nabla \J \leftp \x^{(u)}_0 + \rho^{*}_u \cdot t \cdot \B^{1/2}_0 \cdot \va^{*}_u \rightp^T \\
& \cdot & \Q^{(u)} \cdot \vw^{*} \cdot dt 
\end{eqnarray*}
for any $\x^{(u+1)}$ on the ray $\x^{(u)}+ \rho_u \cdot \BE^{1/2}_0 \cdot \va^{*}_u$, with $\rho_u \in [0,\,1]$, we have
\begin{eqnarray*}
\displaystyle
\J \leftp \x^{(u)}  \rightp - \J \leftp \x^{(u+1)}  \rightp \ge \J \leftp \x^{(u)}  \rightp - \J \leftp \x^{(u)} + \rho^{*}_u \cdot  \BE^{1/2}_0 \cdot \va^{*}_u \rightp\,,
\end{eqnarray*}
hence:
\begin{eqnarray*}
\displaystyle
\J \leftp \x^{(u)}  \rightp &-& \J \leftp \x^{(u+1)}  \rightp \ge  -\rho^{*}_u \cdot \nabla \J \leftp \x^{(u)} \rightp^T \cdot \BE^{1/2}_0 \cdot \va_u^{*} \\
&-& \rho^{*}_u \cdot  \int_{0}^1 \, \leftb \nabla \J \leftp \x^{(u)} + \rho^{*}_u \cdot t \cdot \BE^{1/2}_0 \cdot\va^{*}_u \rightp - \nabla \J \leftp \x^{(u)} \rightp \rightb^T \\
& \cdot & \BE^{1/2} \cdot \va^{*}_u \cdot dt \,,
\end{eqnarray*}
by the Cauchy Schwarz inequality, we have
\begin{eqnarray*}
\displaystyle
\J \leftp \x^{(u)}  \rightp &-& \J \leftp \x^{(u+1)}  \rightp \ge  -\rho^{*}_u \cdot \nabla \J \leftp \x^{(u)} \rightp^T \cdot \BE^{1/2} \cdot \va^{*}_u \\
&-& \rho^{*}_u \cdot \int_{0}^1 \, \norm{ \nabla \J \leftp \x^{(u)} + \rho^{*}_u \cdot t \cdot \BE^{1/2} \cdot \va^{*}_u \rightp - \nabla \J \leftp \x^{(u)} \rightp } \\
& \cdot & \norm{\BE^{1/2} \cdot \va^{*}_u} \cdot dt \\
&\ge &  -\rho^{*}_u \cdot \nabla \J \leftp \x^{(u)} \rightp^T \cdot \BE^{1/2}_0 \cdot \va^{*}_u \\
&-& \rho^{*}_u \cdot \int_{0}^1 \,L \cdot  \norm{  \rho^{*}_u \cdot t \cdot \BE^{1/2}_0 \cdot \va^{*}_u   } \cdot \norm{\BE^{1/2}_0 \cdot \va_u^{*}} \cdot dt \\
&=& - \rho^{*}_u \cdot \nabla \J \leftp \x^{(u)} \rightp^T \cdot \BE^{1/2}_0 \cdot \va^{*}_u \\
&-& \rho^{*}_u \cdot L \cdot \norm{\BE^{1/2}_0 \cdot \va^{*}_u} \cdot \int_{0}^1 \norm{t \cdot \rho^{*}_u \cdot \BE^{1/2}_0 \cdot \va^{*}_u} \cdot dt \\
&=& - \rho^{*}_u \cdot \nabla \J \leftp \x^{(u)} \rightp^T \cdot \BE^{1/2}_0 \cdot \va^{*}_u  - \frac{1}{2} \cdot {\rho^{*}_u}^2 \cdot L \cdot \norm{\BE^{1/2}_0 \cdot \va^{*}_u}^2 \,,
\end{eqnarray*}
since the gradient is Lipschitz continuous, we can choose
\begin{eqnarray*}
\displaystyle
\rho^{*}_u = -\frac{\nabla \J \leftp \x^{(u)} \rightp^T \cdot \BE^{1/2}_0 \cdot \va^{*}_u}{L \cdot \norm{\BE^{1/2}_0 \cdot \va^{*}_u}^2} \,,
\end{eqnarray*}
therefore,
\begin{eqnarray*}
\displaystyle
\J \leftp \x^{(u)}  \rightp &-& \J \leftp \x^{(u+1)}  \rightp \ge \frac{\leftb \nabla \J \leftp \x^{(u)} \rightp^T \cdot \BE^{1/2}_0 \cdot \va^{*}_u \rightb^2}{L \cdot \norm{\BE^{1/2}_0 \cdot \va^{*}_u}^2} \\
&-& \frac{1}{2} \cdot \frac{\leftb- \nabla \J \leftp \x^{(u)} \rightp^T \cdot \BE^{1/2}_0 \cdot \va^{*}_u \rightb^2}{L \cdot \norm{\BE^{1/2}_0 \cdot \va^{*}_u}^2} \\
 &=& \frac{1}{2 \cdot L} \cdot \leftb- \frac{ \nabla \J \leftp \x^{(u)} \rightp^T \cdot \BE^{1/2}_0 \cdot \va^{*}_u }{ \norm{\BE^{1/2}_0 \cdot \va^{*}_u}} \rightb^2 \,.
\end{eqnarray*}
By \eqref{cond:cond1}, and \eqref{eq:theorem-ortho}, it follows that $\leftl \J \leftp \x^{(u)}\rightp\rightl_{u = 0}^{\infty}$ is a monotone decreasing number sequence and it has a bound below, therefore $\leftl \J \leftp \x^{(u)}\rightp\rightl_{u = 0}^{\infty}$ has a limit, and consequently \eqref{eq:to-prove} holds.
\end{proof}

\section{Experimental Settings}
\label{sec:experimental-settings}

We make use of the Lorenz-96 model \cite{LorenzModel} as our surrogate model in order to tests the compared EnKF implementations. The Lorenz-96 model  is described by the following set of ordinary differential equations \cite{fertig2007comparative}:
\begin{eqnarray}
\label{eq:Lorenz-model}
\displaystyle
\frac{dx_j}{dt} = \begin{cases}
\leftp x_2-x_{\Nstate-1}\rightp \cdot x_{\Nstate} -x_1 +F & \text{ for $j=1$}, \\
\leftp x_{j+1}-x_{j-2}\rightp \cdot x_{j-1} -x_j +F & \text{ for $2 \le j \le \Nstate-1$}, \\
\leftp x_1-x_{\Nstate-2}\rightp \cdot x_{\Nstate-1} -x_{\Nstate} +F & \text{ for $j=\Nstate$}, 
\end{cases}
\end{eqnarray}
where $F$ is an external force and $\Nstate=40$ is the number of model components. Periodic boundary conditions are assumed. When $F=8$ units, the model exhibits chaotic behavior, which makes it a relevant surrogate problem for atmospheric dynamics \cite{karimi2010extensive,gottwald2005testing}. A time unit in the Lorenz-96 represents 120 hours (5 days) in the atmosphere. 

The experimental settings are described below, they are similar to those presented in \cite{nerger2014influence}:
\begin{itemize}
\item An initial random solution is integrated over a long time period in order to obtain an initial condition $\x^{*}_{-2} \in \Re^{\Nstate \times 1}$ dynamically consistent with the model \eqref{eq:Lorenz-model}.
\item A perturbed background solution $\x^b_{-2}$ is obtained at time $t_{-2}$ by drawing a sample from the Normal distribution,
\begin{eqnarray*}
\displaystyle
\x^b_{-2} \sim \Nor \leftp \x^{*}_{-2},\, 0.05^2 \cdot \I  \rightp \,,
\end{eqnarray*}
this solution is then integrated for 10 time units (equivalent to 50 days in the atmosphere) in order to obtain a background solution $\x^b_{-1}$ consistent with the numerical model.
\item An initial perturbed ensemble is built about the background state by taking samples from the distribution,
\begin{eqnarray*}
\displaystyle
\x^{b[i]}_{-1} \sim \Nor \leftp \x^{b}_{-1},\, 0.05^2 \cdot \I  \rightp \,, \text{ for $1 \le i \le \Nens$} \,.
\end{eqnarray*}
These samples are propagated over a time period of 10 units in order to make them consistent with the model dynamics \eqref{eq:Lorenz-model}, after which the initial ensemble members $\x^{b[i]}_0 \in \Re^{\Nstate \times 1}$ are obtained. We build an initial pool of members. Two-dimensional projections of such collection are shown in figure \ref{fig:initial-pool}.
\begin{figure}[H]
    \centering
    \begin{subfigure}[b]{0.5\textwidth}
        \centering
        \includegraphics[height=1.2in]{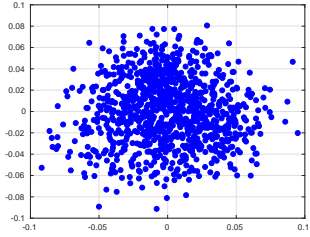}
        \caption{2D projection of the inital members.}
    \end{subfigure}%
    ~ 
    \begin{subfigure}[b]{0.5\textwidth}
        \centering
        \includegraphics[height=1.25in]{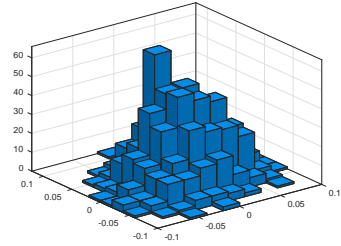}
        \caption{2D initial background error distribution.}
    \end{subfigure}
    \caption{Initial pool of background members for the experiments. Two dimensional projections based on the two leading directions are shown.}
    \label{fig:initial-pool}
\end{figure}
\item The number of assimilation steps reads $\N=500$. Observations are collected every 12 hours and their error statistics are described by the Gaussian distribution. We use the Dormand and Prince method \cite{dormand1980family} with absolute error tolerance of $10^{-7}$ in order to propagate the numerical model every 0.5 time units,
\begin{eqnarray}
\label{eq:observation-settings}
\displaystyle
\y_k \sim \Nor \leftp \H_k \cdot \x^{*}_k\,, 0.01^2 \cdot \I \rightp \,, \text{ for $1 \le k \le \N$} \,.
\end{eqnarray}
We try different number of observations per assimilation step, these are given by
\begin{eqnarray*}
\displaystyle 
\Nobs = p \cdot \Nstate \,,
\end{eqnarray*}
where $p = \{70 \% \,,100 \% \}$, which corresponds with $m \in \{28,\,40\}$ observations.  Observations from the model state are randomly chosen at the different assimilation steps. We employ the non-linear observation operator \cite{van2015nonlinear}:
\begin{eqnarray}
\label{eq:operator-non-smooth}
\displaystyle
\leftl \Ho \leftp \x \rightp \rightl_j = \frac{\leftl \x \rightl_j }{2} \cdot \leftb \leftp\frac{\lefta \leftl \x \rightl_j \righta}{2} \rightp^{\gamma-1} +1 \rightb \,,
\end{eqnarray}
where $j$ denotes the $j$-th observed component from the model state. Likewise, $1 \le \gamma \le 7$. Note that, the experiments are performed under perfect model assumptions.
\item We evenly range the inflation factor $\rho \in [1.1,\, 1.9]$ with increments of $0.1$.
\item Different ensemble sizes are tried during the experiments, $\Nens \in \{20,\,40\}$.
\item The L--2 norm of errors is utilized as a measure of accuracy, at the assimilation step $k$, it is defined as follows:
\begin{eqnarray}
\label{eq:errors-L2}
\lambda_k = \norm{\x^{*}_k - \xm^{a}_k}_2 = \sqrt{\leftb \x^{*}_k - \xm^{a}_k \rightb^T \cdot \leftb \x^{*}_k - \xm^{a}_k \rightb} \,,
\end{eqnarray}
where $\x^{*}_k$ and $\x^a_k$ are the reference and the analysis solutions, respectively.
\item The Root-Mean-Square-Error (RMSE) measures the performance of a filter, in average, for a given number of assimilation steps $\N$:
\begin{eqnarray*}
\displaystyle
\lambda = \sqrt{\frac{1}{\N} \cdot \sum_{k=1}^{\N} \lambda_k^2} \,.
\end{eqnarray*}
\item The number of optimization steps $U$ equals 10. We choose this value for two main reasons: model propagations are involved during optimization steps and even more, for either of the compared filter implementations, the 4D-EnKF-MC or the 4D-MLEF, after 10 iterations no reduction in cost-function values can be evidenced. 
\end{itemize}

For each configuration of the parameters, we perform 30 experiments to get insights about the error behavior and his associate statistics. We employ the first two statistical moments (mean and variance) of underlying error distributions to analyze the numerical results. The analysis is divided into two groups, (nearly) linear observation operators where $1 \le \gamma \le 3$ and highly non-linear observation operators $4 \le \gamma \le 7$.

\subsection{Results with (non-) linear observation operators $1 \le \gamma  \le 3$}

For different values of $\Nens$ and $1 \le \gamma \le 3$, the results are shown in figures \ref{fig:1-3-70} and \ref{fig:1-3-100} for $p=70\%$ and $p=100\%$ (full observational networks), respectively. We set $\ra = 2$ to build the set of basis vectors in the 4D-EnKF-MC context. As expected, the accuracy of both filter implementations, the 4D-EnKF-MC, and the 4D-MLEF, are better than that of the NO Data Assimilation (NODA) forecast, which is reported as well. Note that increments in the ensemble size have a clear impact on the 4D-EnKF-MC performance; the larger the ensemble size, the better (in terms of accuracy) the posterior estimate, for instance, the L--2 error norms are decreased by some order of magnitudes. This behavior is not evident in the 4D-MLEF method wherein slight improvements can be evidenced. Note that both filter implementations benefit from having more observations to digest during assimilation steps. This is, errors can be decreased as the number of observations is increased if the assimilation window keeps constant. Again, this patron is more evident in the 4D-EnKF-MC method than in the 4D-MLEF filter.  This can be explained as follows: during assimilation steps, the 4D-EnKF-MC method builds control spaces whose dimensions equal the model one. Thus, all information brought by the observations can be captured onto that space (whose dimension is larger than the observation one), not the 4D-MLEF whose space dimension relies on the ensemble size.

\subsection{Results with highly non-linear observation operators $4 \le \gamma \le 7$}
For different values of $\Nens$ and $4 \le \gamma \le 7$, the results are shown in figures \ref{fig:4-7-70} and \ref{fig:4-7-100} for $p=70\%$ and $p=100\%$ (full observational networks), respectively. Again, we set $\ra = 2$ to build the set of basis vectors in the 4D-EnKF-MC context. As expected, all filter implementations provide better results than those obtained by the NODA forecast. In all cases, the results proposed by the 4D-EnKF-MC are better than those of the 4D-MLEF in terms of L--2 error norms. For $p = 70 \%$, the results of the compared filter implementations get closer as the degree $\gamma$ increases. Nevertheless, as the number of observations increases, the 4D-EnKF-MC provides better results than those of the 4D-MLEF. Note that, in the 4D-EnKF-MC, the quality of analysis corrections improves as the number of observations increases. Increments in the ensemble size have a clear impact in the quality of both filters, the larger the ensemble size, the better the quality of analysis corrections.

\subsection{Further discussions}

For all parameter configurations, the results are shown in the Tables \ref{tab:rho-0.7} and \ref{tab:rho-1} for $p = 70 \%$ and $p = 100 \%$, respectively. In the 4D-EnKF-MC method, besides $\ra = 2$, the values of 6 and 18 are tried as the radius of influence for building the set of basis vectors during assimilation steps, in this manner, we can have an idea about how the radius size impacts the quality of the initial analysis increments. The first thing to note is that, as the radius of influence increases, the initial analysis corrections' quality degrades. This can be expected given the special form of the error dynamics in the numerical model \eqref{eq:Lorenz-model} wherein fast decay of error correlations can be possible. For instance, the model dynamics strongly correlate errors among neighboring components. Thus, having large radii of influences can impact the quality of posterior initial members. Note that the behavior of 4D-MLEF is similar for all $\gamma$ values, while for the 4D-EnKF-MC RMSE, values degrade as the $\gamma$ increases. However, for all scenarios, the results of the 4D-EnKF-MC are better than those of the 4D-MLEF. In both filter implementations, their performance can be improved by increasing the number of ensemble members. For instance, in the 4D-MLEF, the larger the ensemble size, the more degrees of freedom are available during optimization steps. In the 4D-EnKF-MC method, under such circumstances, more information about the error dynamics is available to fit the models \eqref{eq:fitting-models-of-the-form}.

\begin{figure}[H]
\centering
\begin{tabular}{ccc} \\ \hline
  & $\Nens = 20$ & $\Nens = 60$   \\ \hline
 \raisebox{-16mm}{\rotatebox[origin=c]{90}{$\gamma=1$}} &    \raisebox{-\totalheight}{\includegraphics[width=0.40\textwidth]{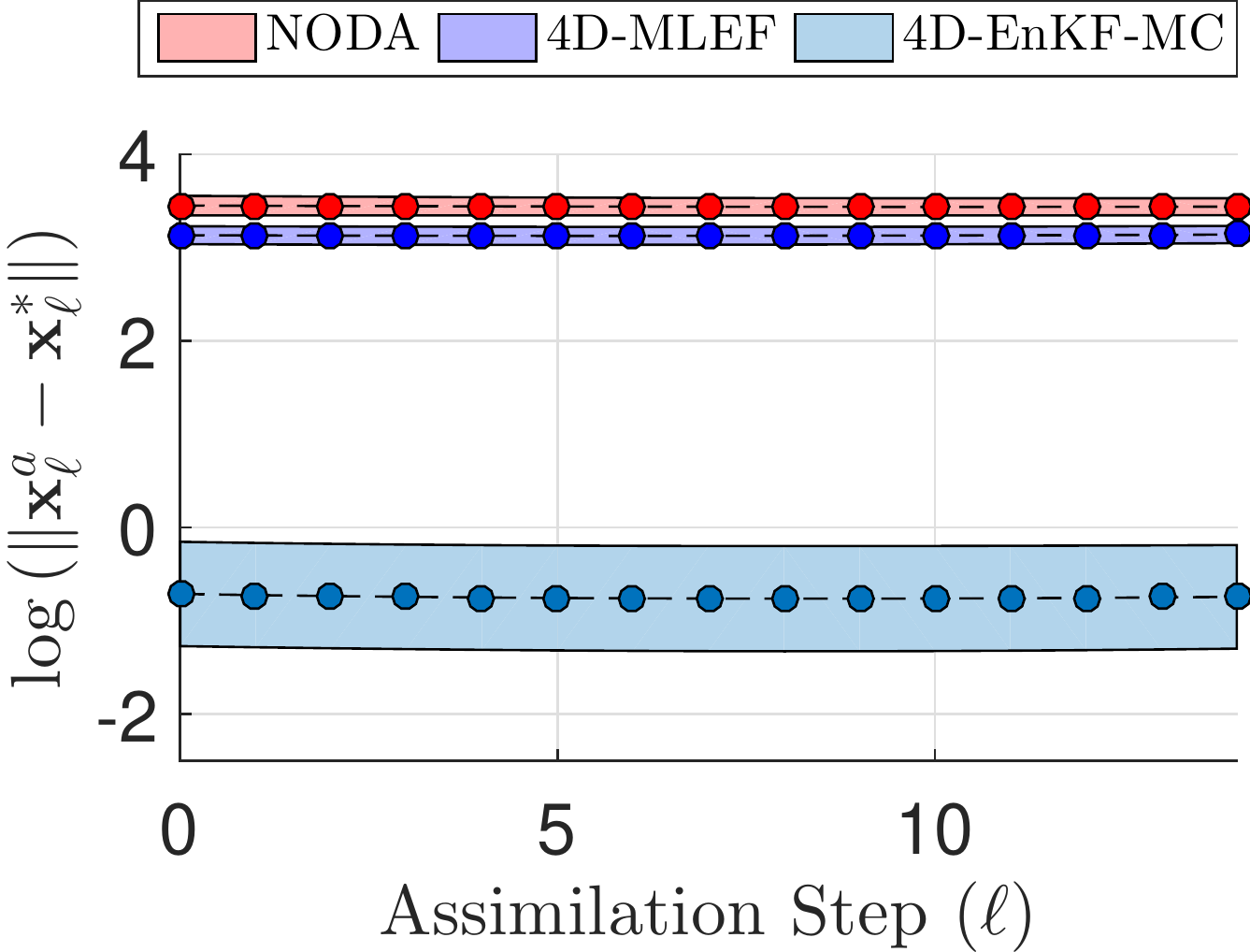}} &    \raisebox{-\totalheight}{\includegraphics[width=0.40\textwidth]{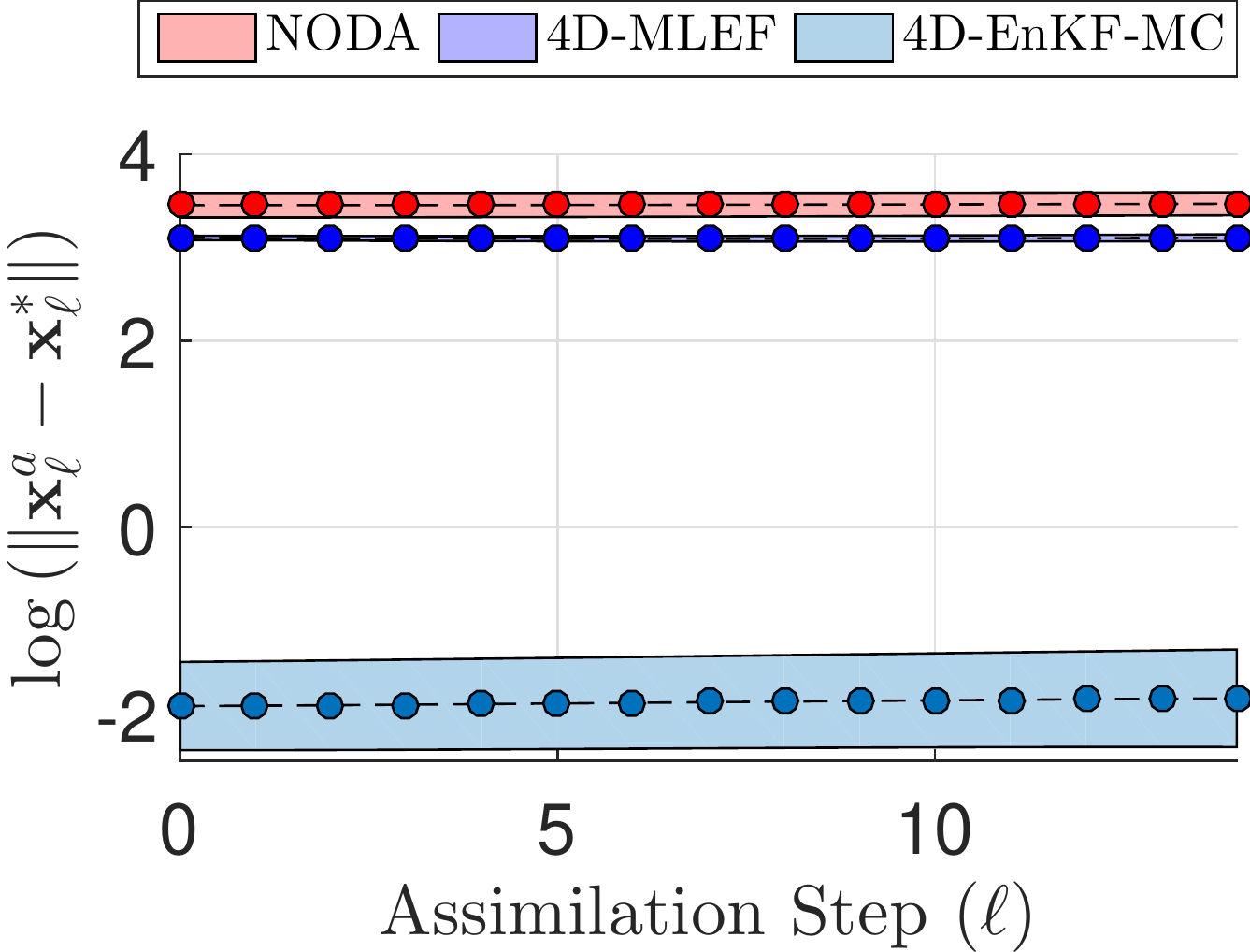}} \\ \hline
 \raisebox{-16mm}{\rotatebox[origin=c]{90}{$\gamma=2$}} &    \raisebox{-\totalheight}{\includegraphics[width=0.40\textwidth]{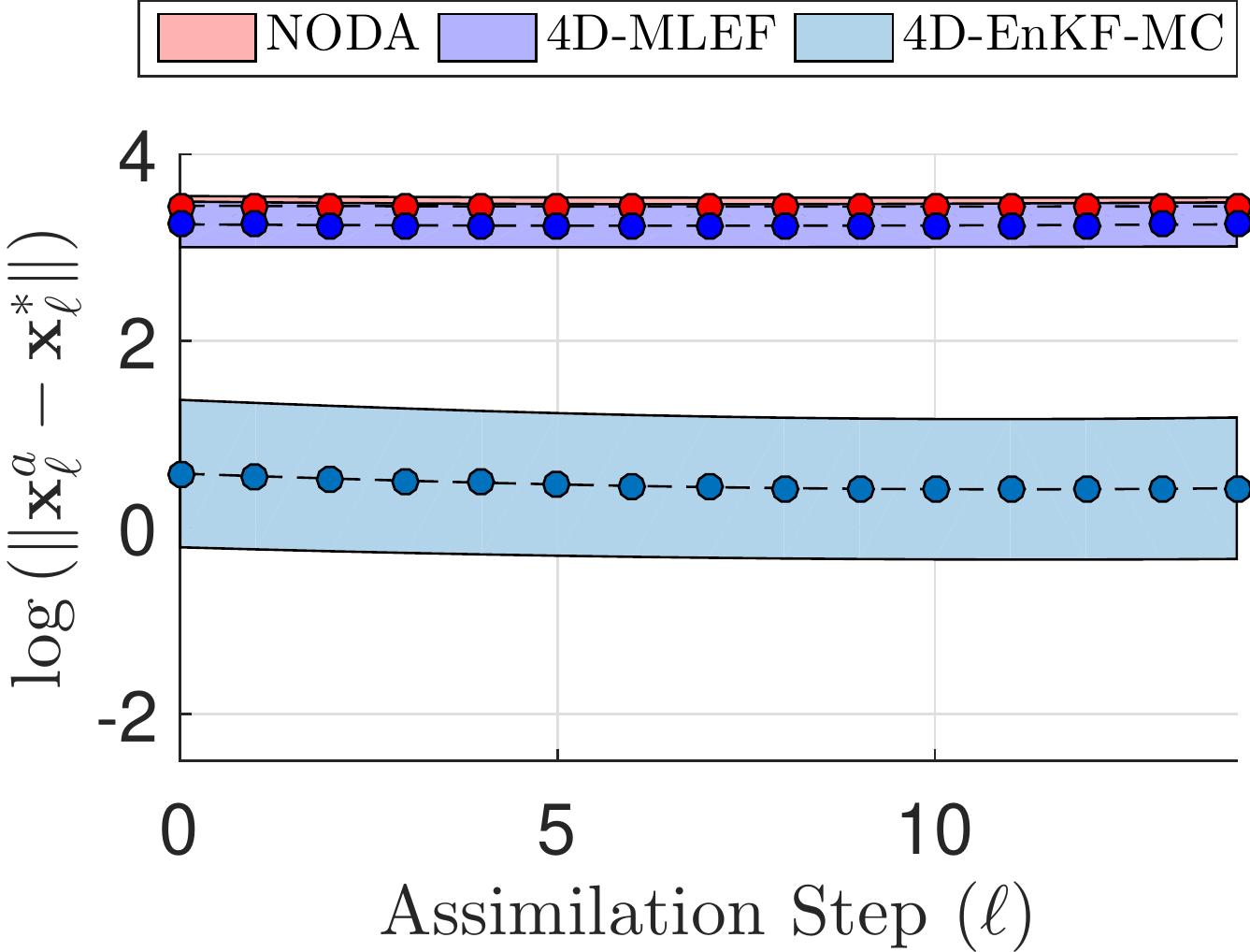}} &    \raisebox{-\totalheight}{\includegraphics[width=0.40\textwidth]{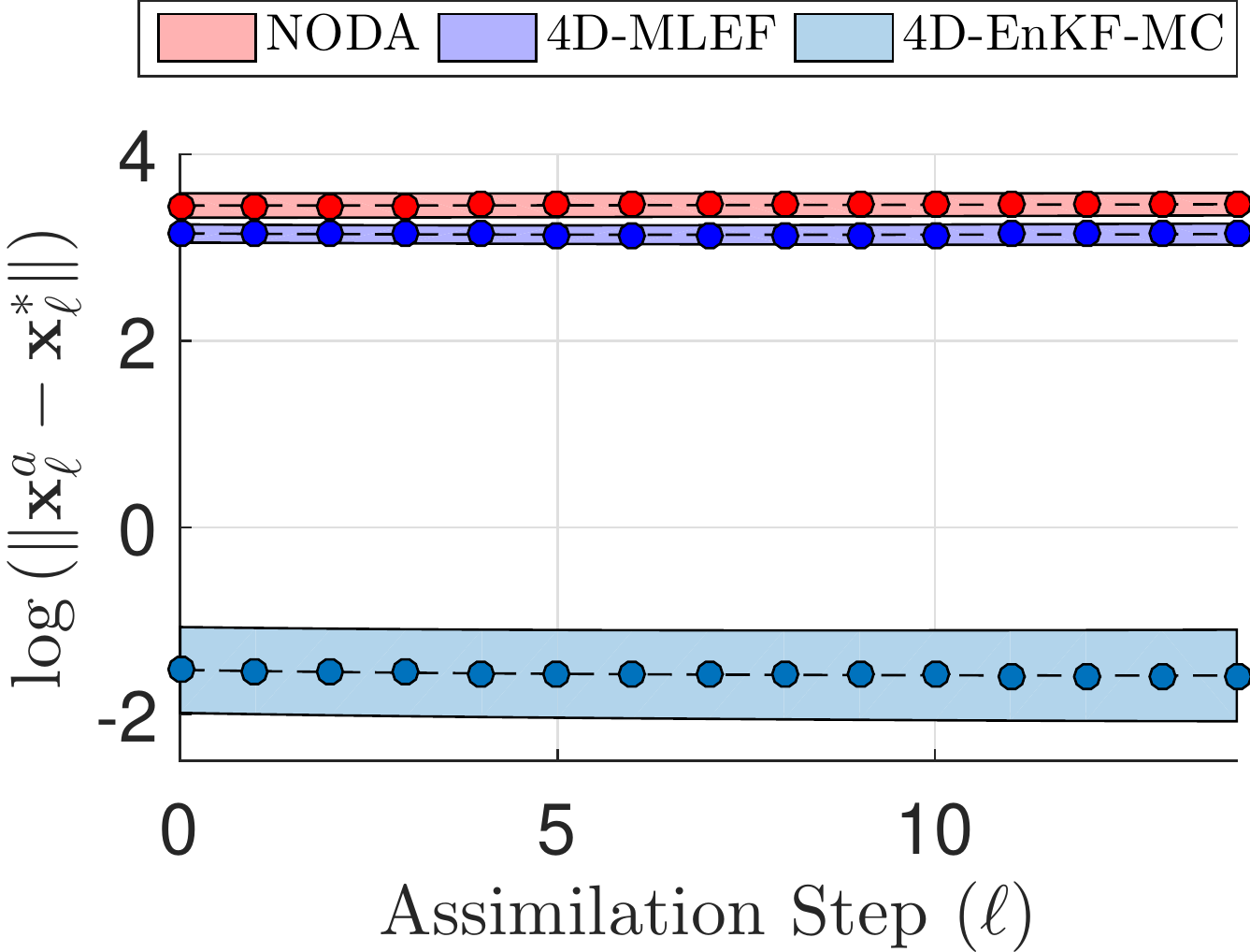}} \\ \hline
  \raisebox{-16mm}{\rotatebox[origin=c]{90}{$\gamma=3$}} &    \raisebox{-\totalheight}{\includegraphics[width=0.40\textwidth]{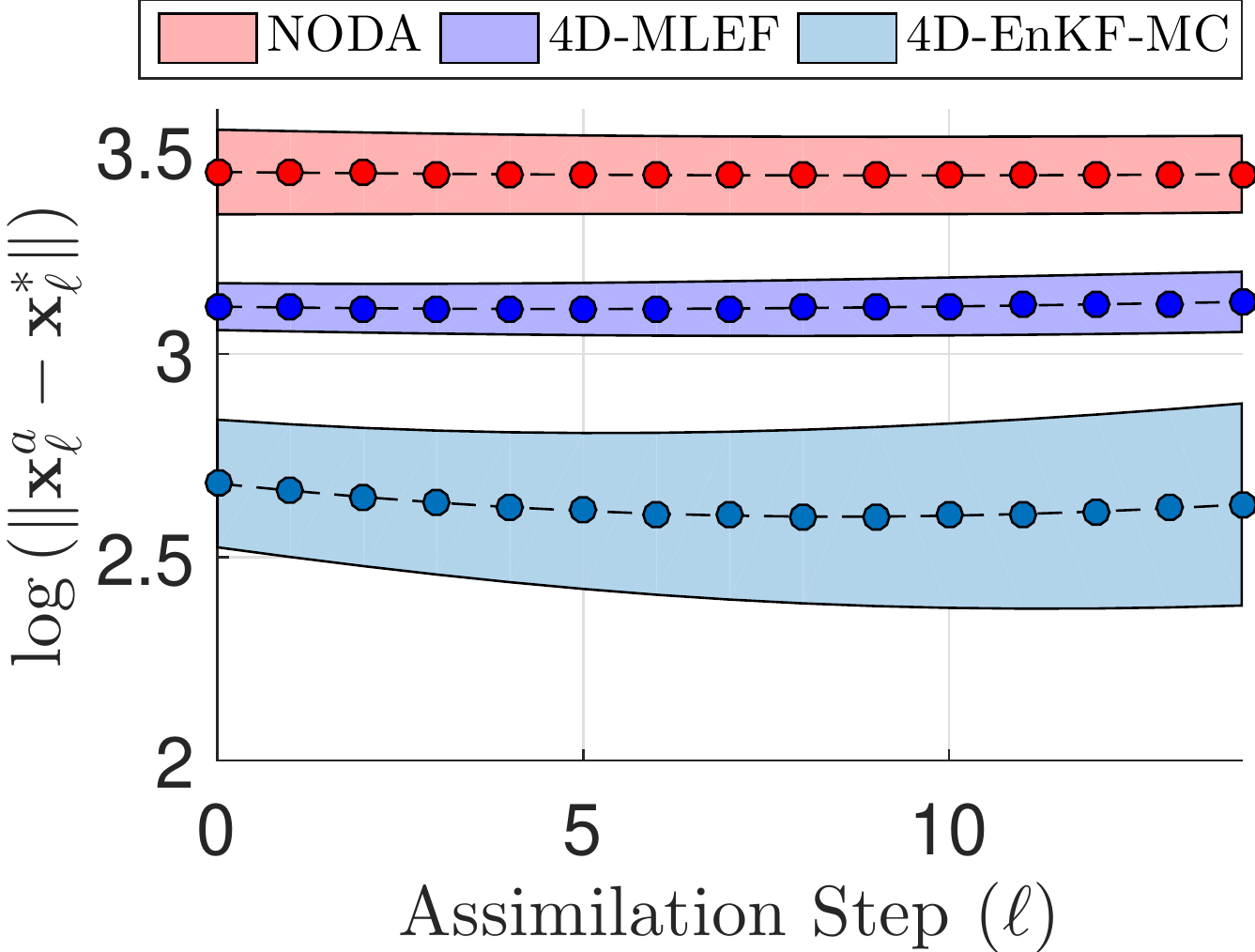}} &    \raisebox{-\totalheight}{\includegraphics[width=0.40\textwidth]{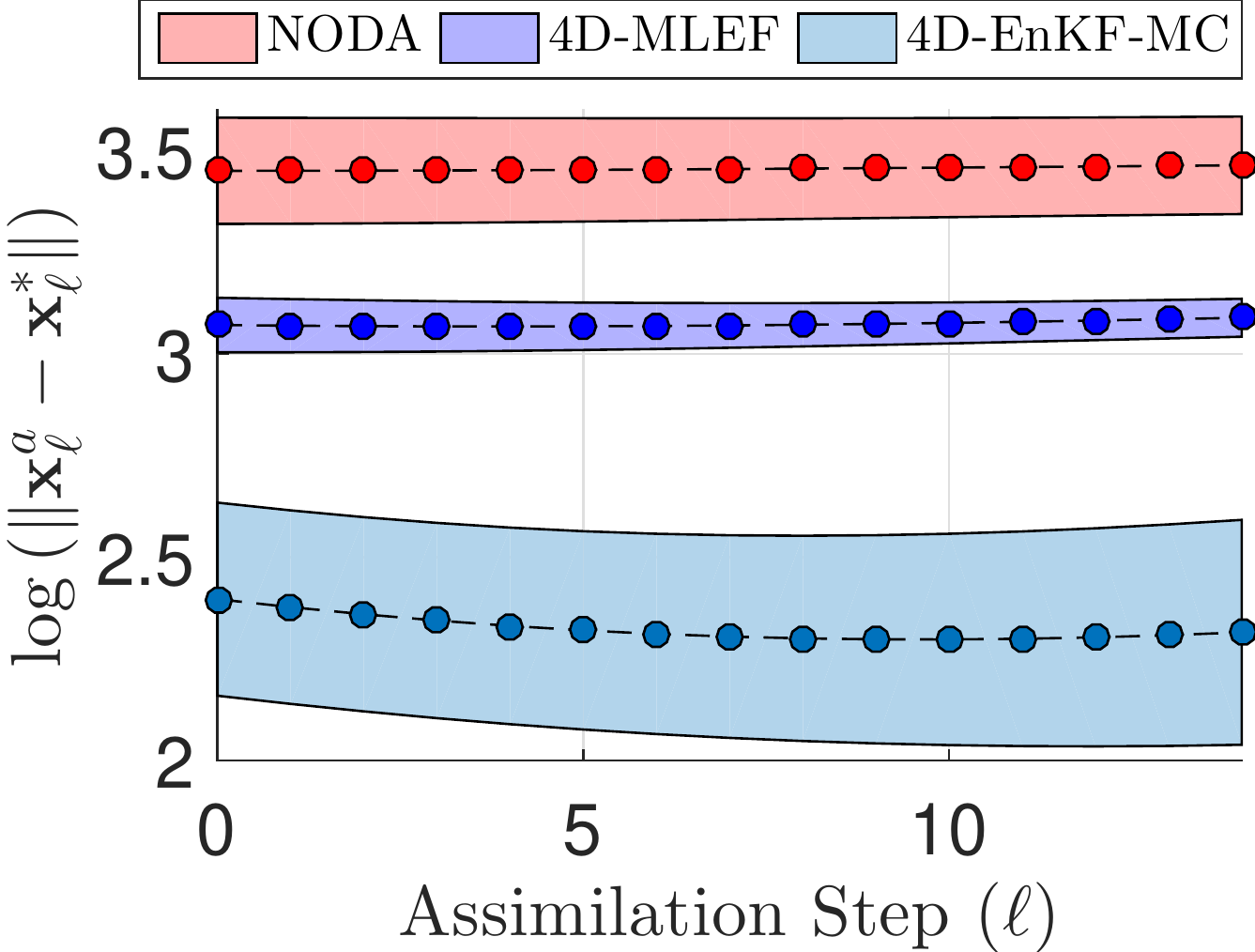}} \\ \hline
\end{tabular}
\caption{Averages (dashed lines) and standard deviations (shaded regions) of error norms for $p = 70\%$. We let $\Nens \in \{20,\,60\}$. Likewise, $\gamma \in \{1,\, 2,\, 3\}$.}
\label{fig:1-3-70}
\end{figure}
\begin{figure}[H]
\centering
\begin{tabular}{ccc} \\ \hline
  & $\Nens = 20$ & $\Nens = 60$   \\ \hline
 \raisebox{-16mm}{\rotatebox[origin=c]{90}{$\gamma=1$}} &    \raisebox{-\totalheight}{\includegraphics[width=0.40\textwidth]{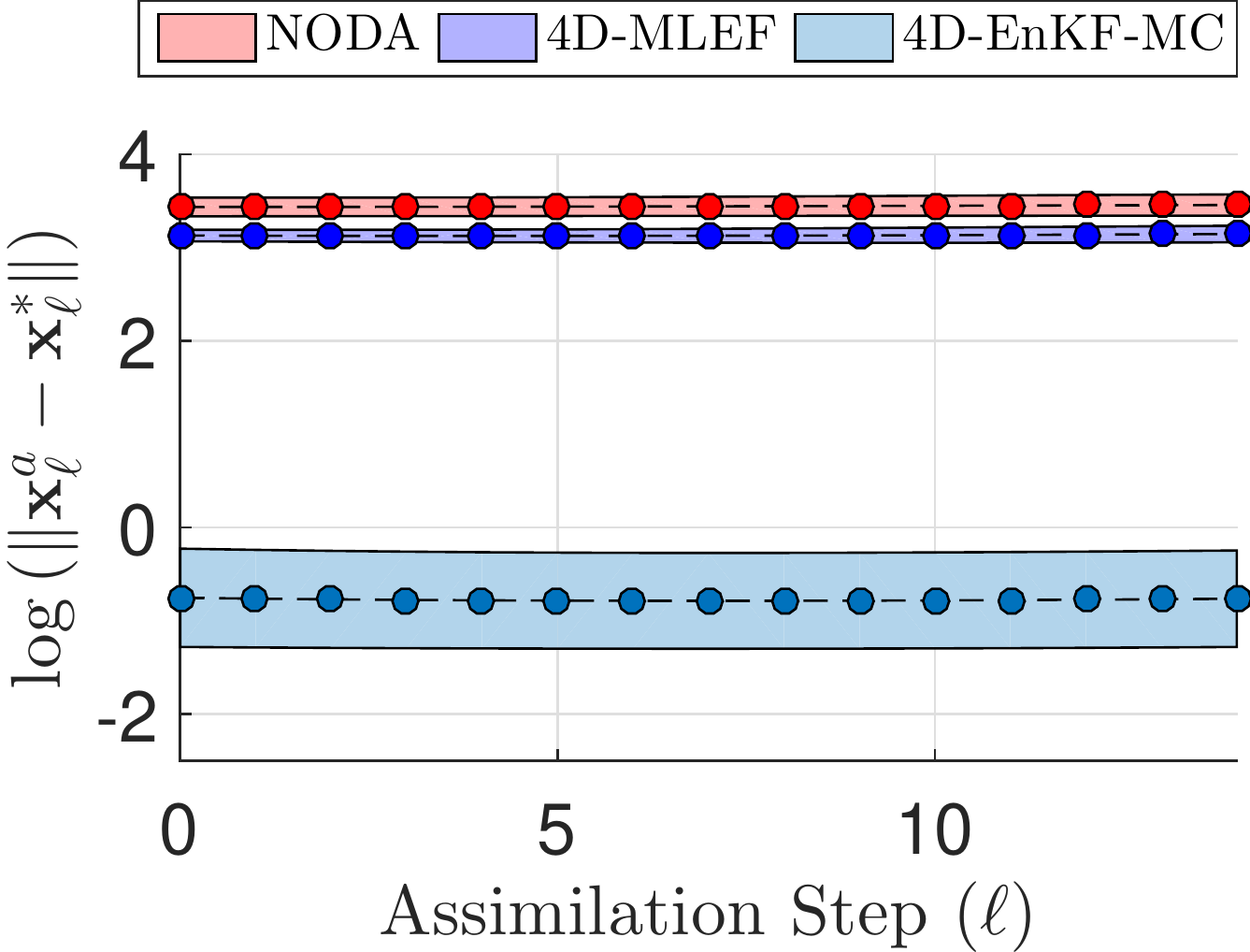}} &    \raisebox{-\totalheight}{\includegraphics[width=0.40\textwidth]{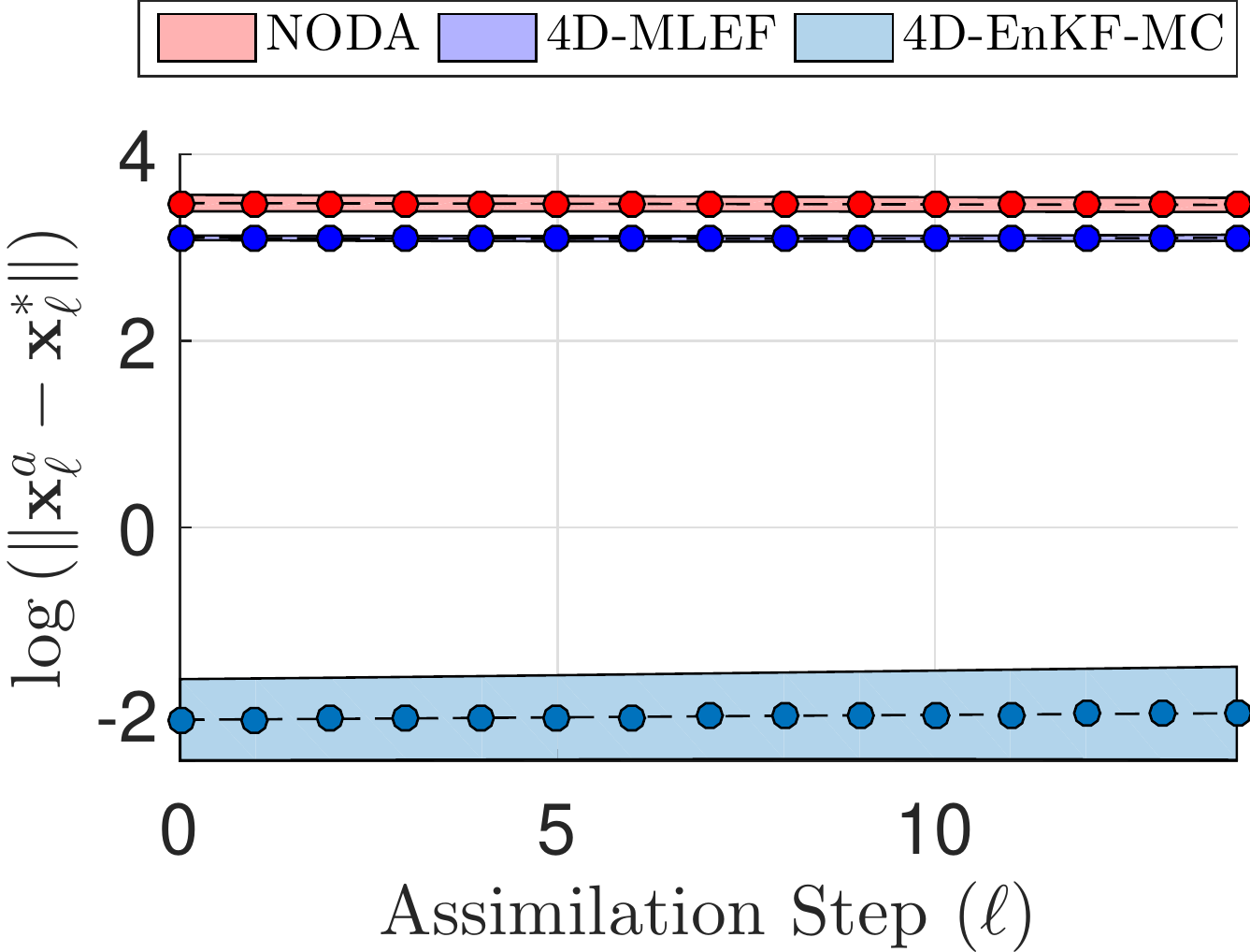}} \\ \hline
 \raisebox{-16mm}{\rotatebox[origin=c]{90}{$\gamma=2$}} &    \raisebox{-\totalheight}{\includegraphics[width=0.40\textwidth]{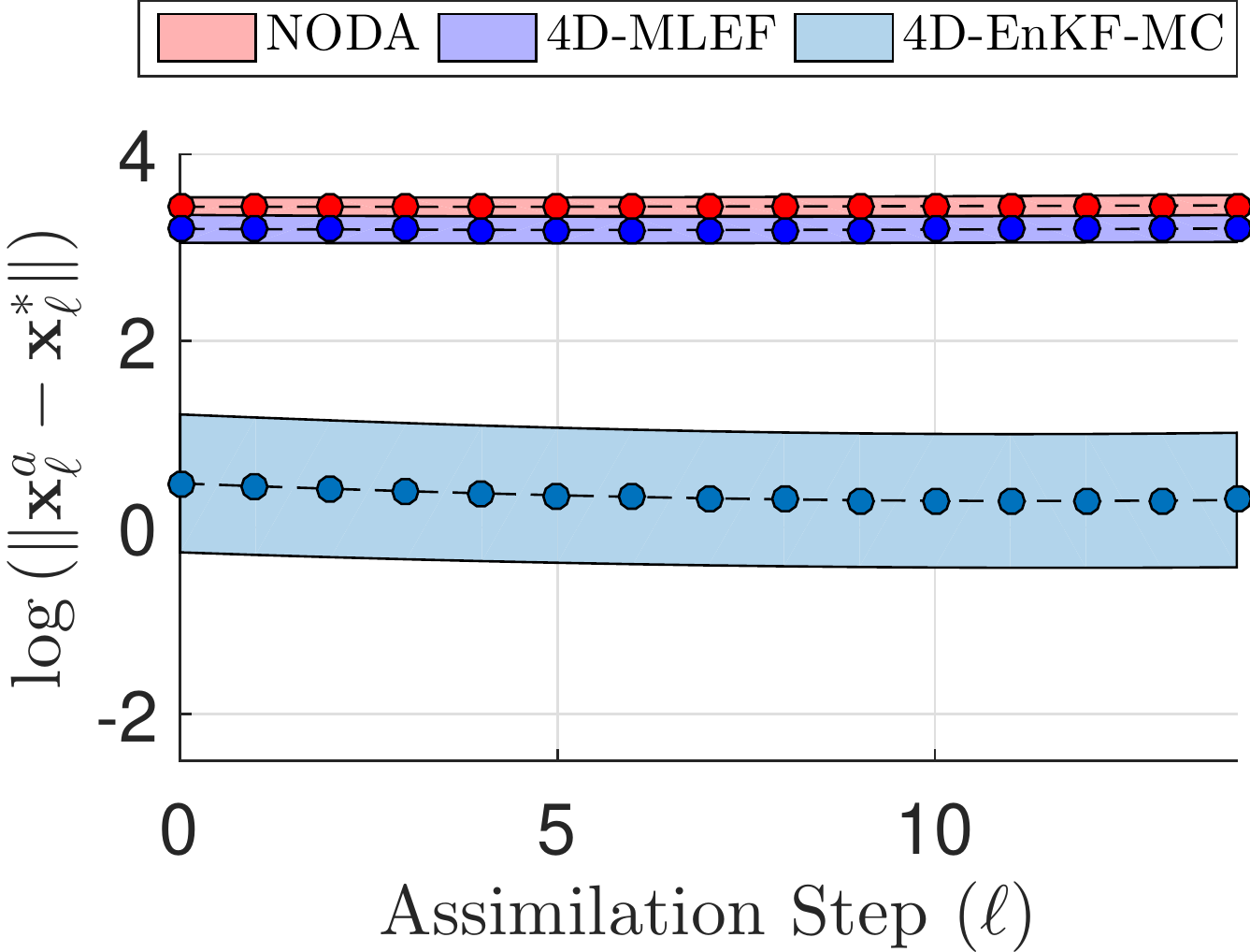}} &    \raisebox{-\totalheight}{\includegraphics[width=0.40\textwidth]{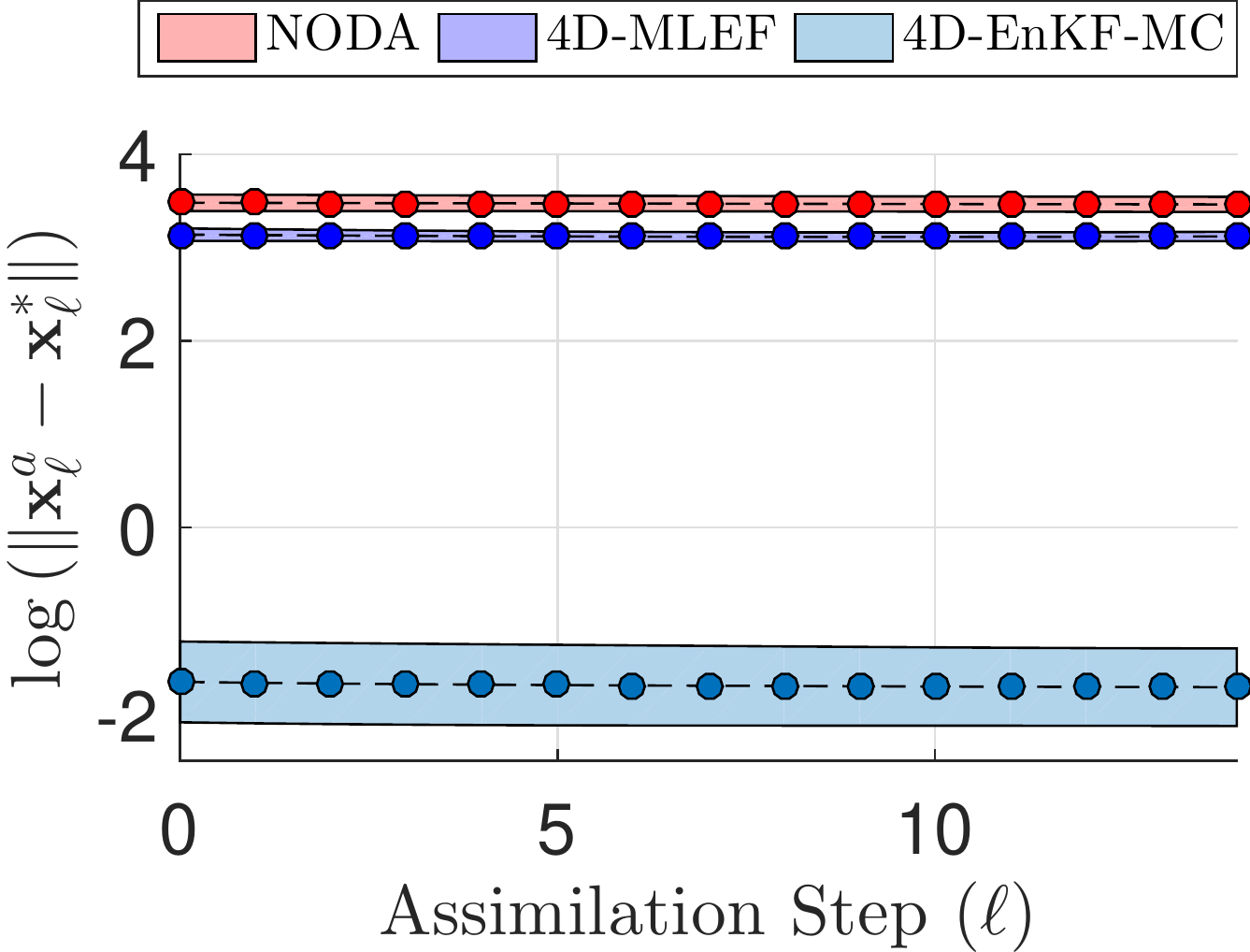}} \\ \hline
  \raisebox{-16mm}{\rotatebox[origin=c]{90}{$\gamma=3$}} &    \raisebox{-\totalheight}{\includegraphics[width=0.40\textwidth]{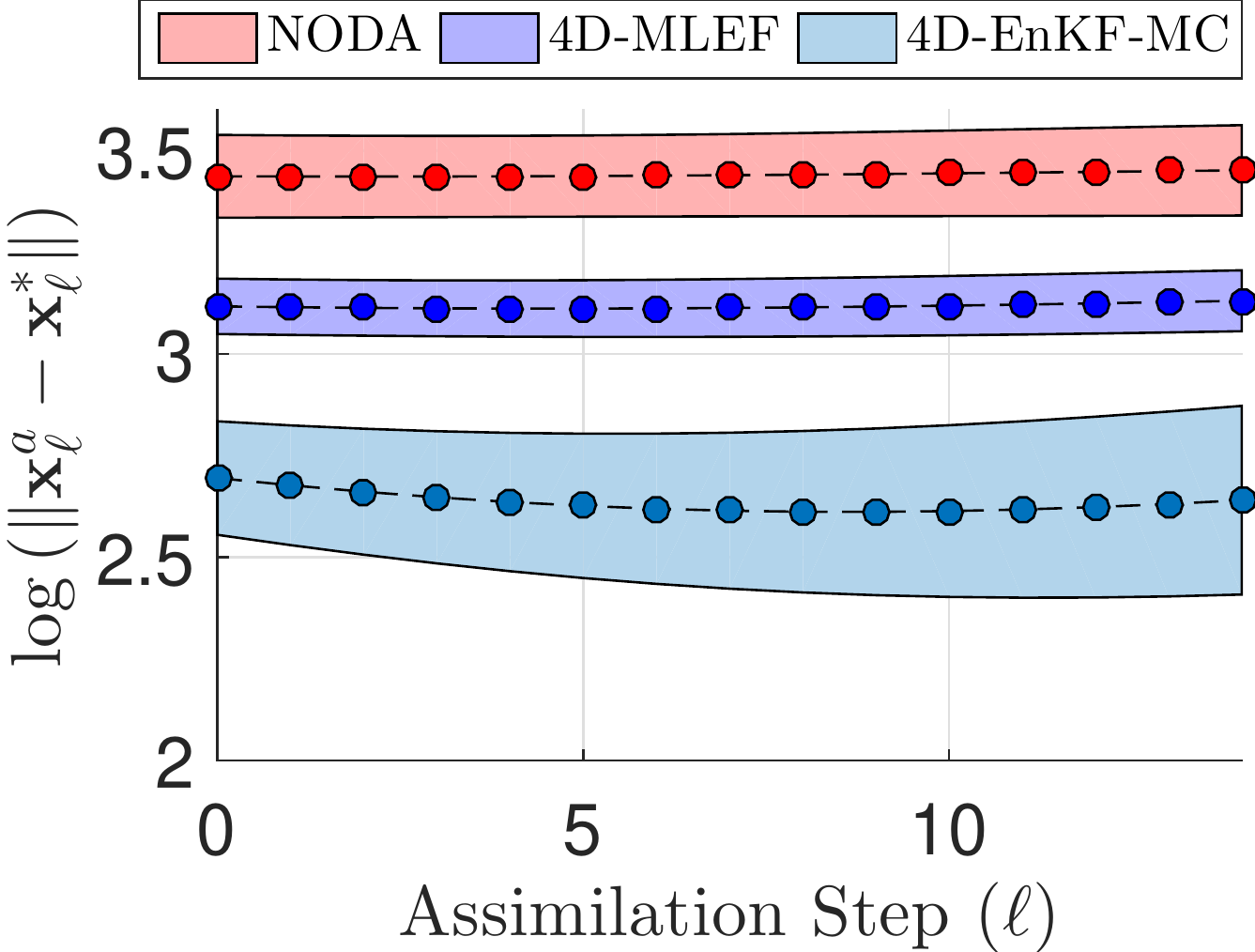}} &    \raisebox{-\totalheight}{\includegraphics[width=0.40\textwidth]{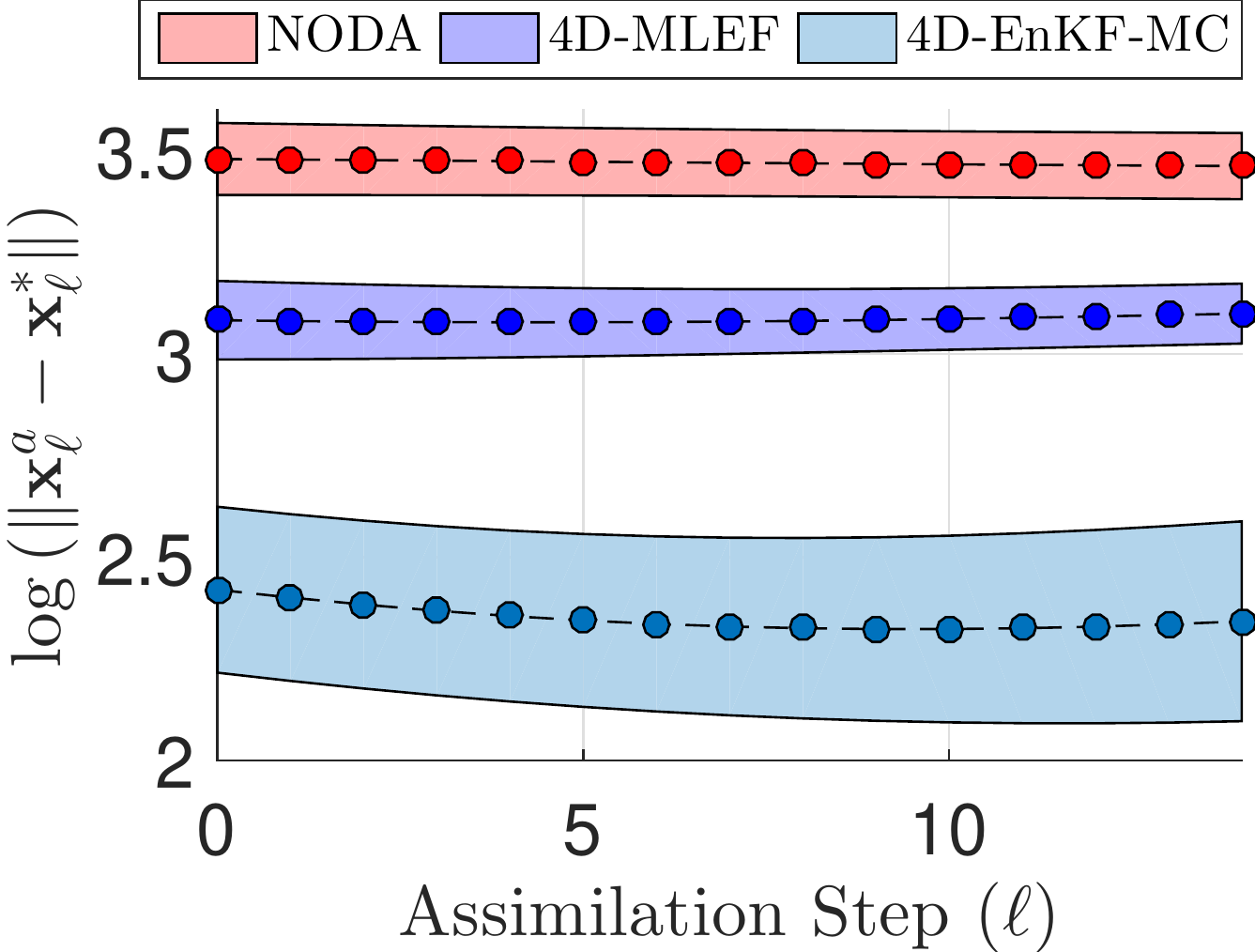}} \\ \hline
\end{tabular}
 \caption{Averages (dashed lines) and standard deviations (shaded regions) of error norms for 30 realizations of experiments and the values of parameters $p = 100 \%, \Nens \in \{20,\,60\}$, and $1 \le \gamma \le 3$.}
\label{fig:1-3-100}
\end{figure}

\begin{figure}[H]
\centering
\begin{tabular}{ccc} \\ \hline
  & $\Nens = 20 \%$ & $\Nens = 60 \%$   \\ \hline
 \raisebox{-16mm}{\rotatebox[origin=c]{90}{$\gamma=4$}} &    \raisebox{-\totalheight}{\includegraphics[width=0.40\textwidth]{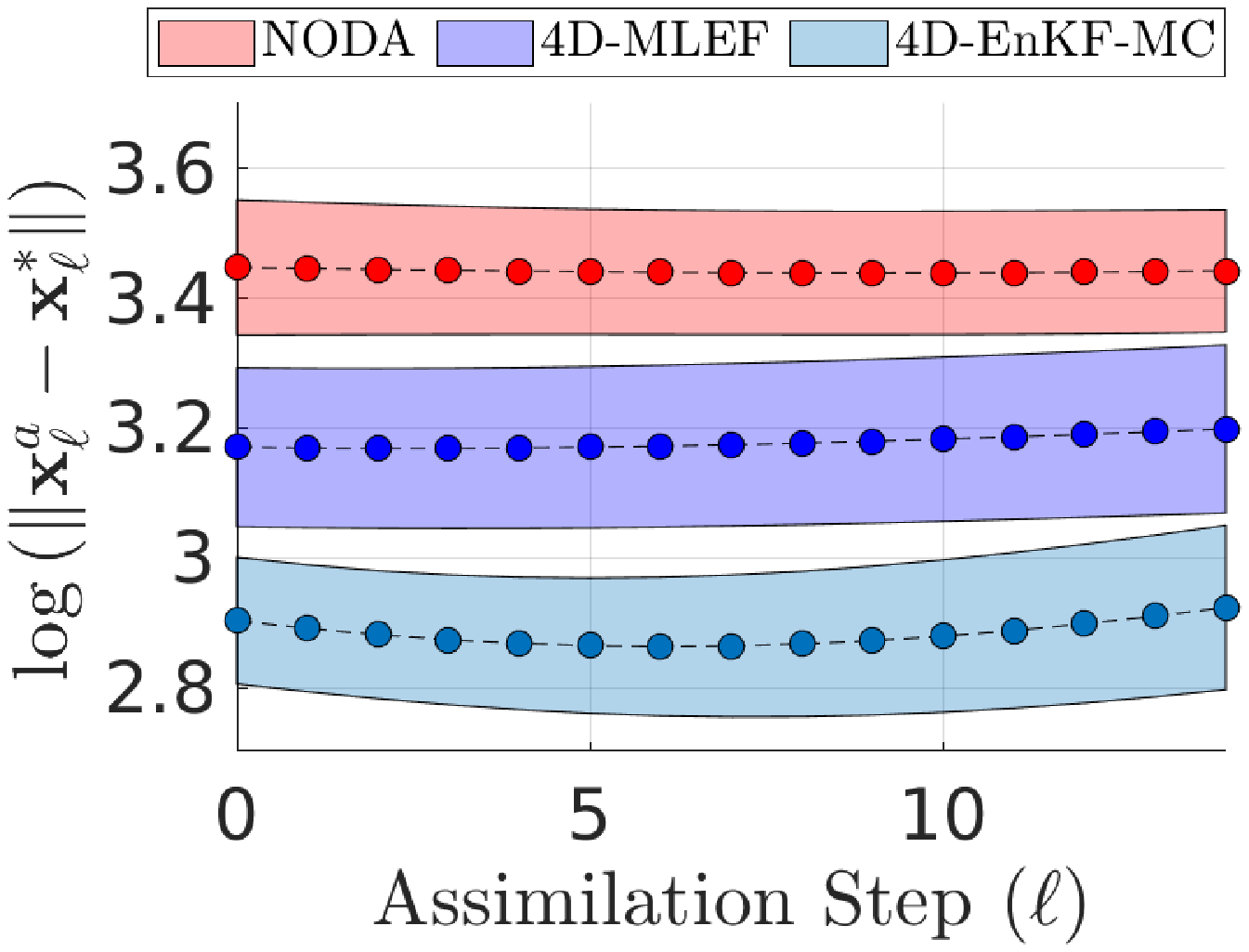}} &    \raisebox{-\totalheight}{\includegraphics[width=0.40\textwidth]{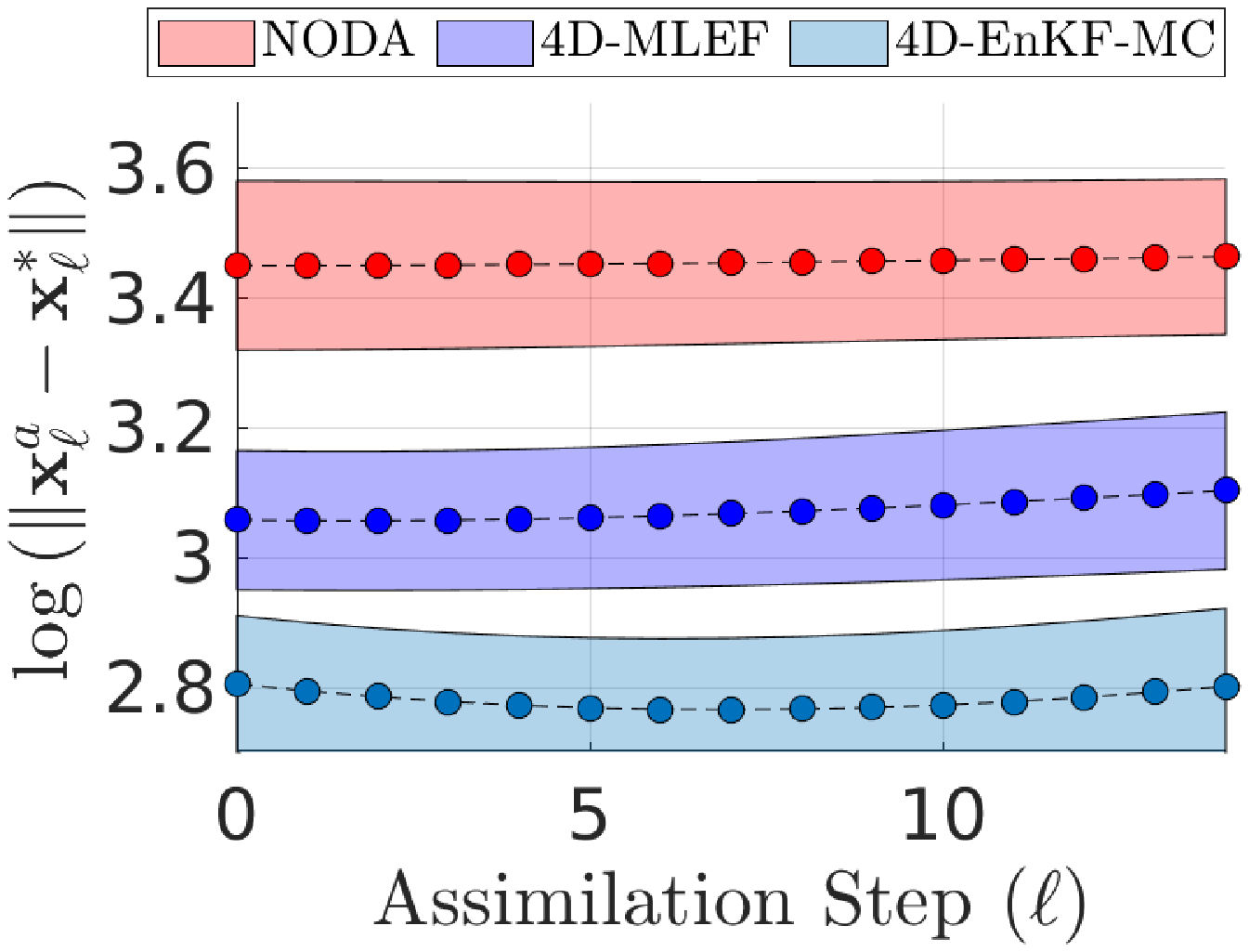}} \\ \hline
 \raisebox{-16mm}{\rotatebox[origin=c]{90}{$\gamma=5$}} &    \raisebox{-\totalheight}{\includegraphics[width=0.40\textwidth]{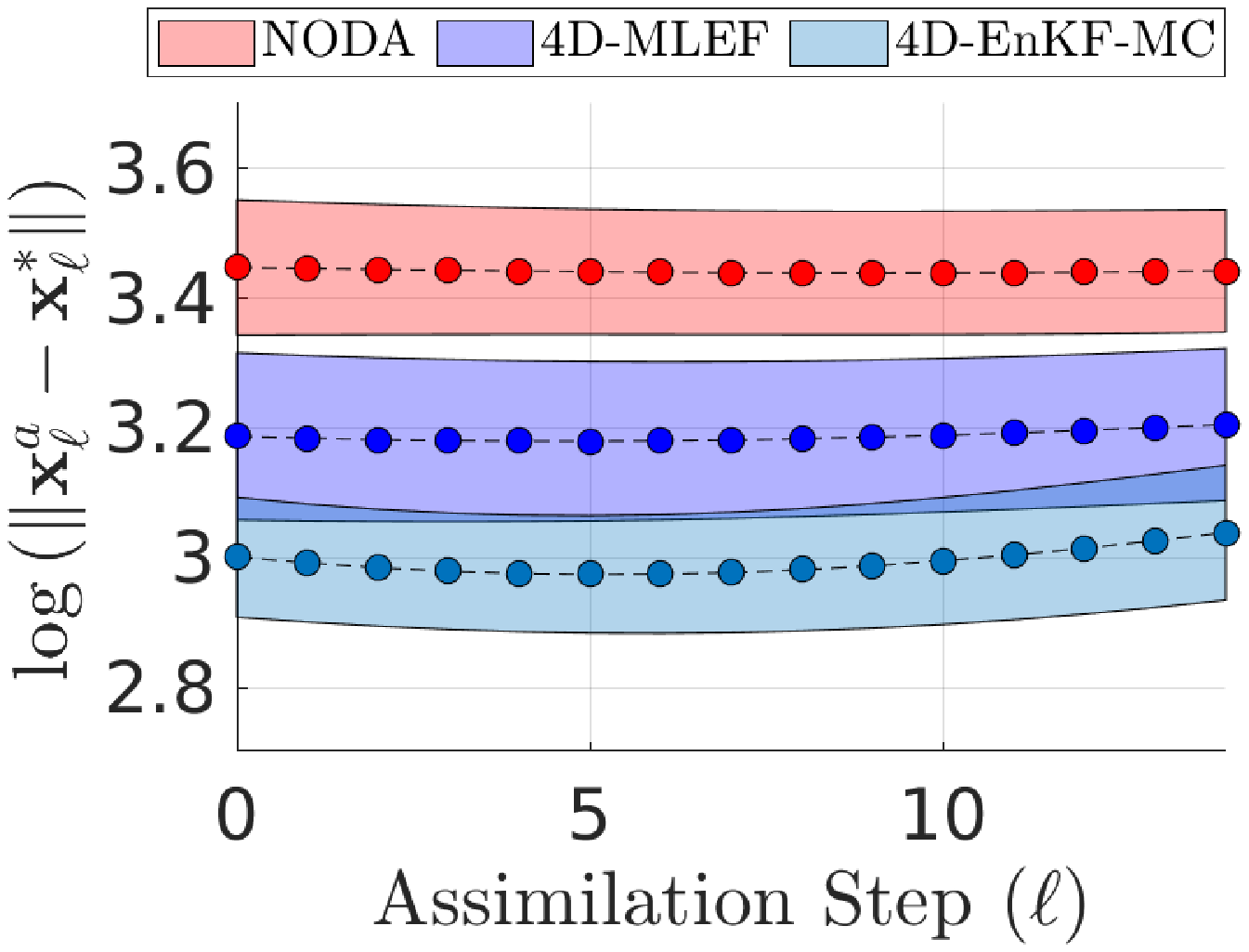}} &    \raisebox{-\totalheight}{\includegraphics[width=0.40\textwidth]{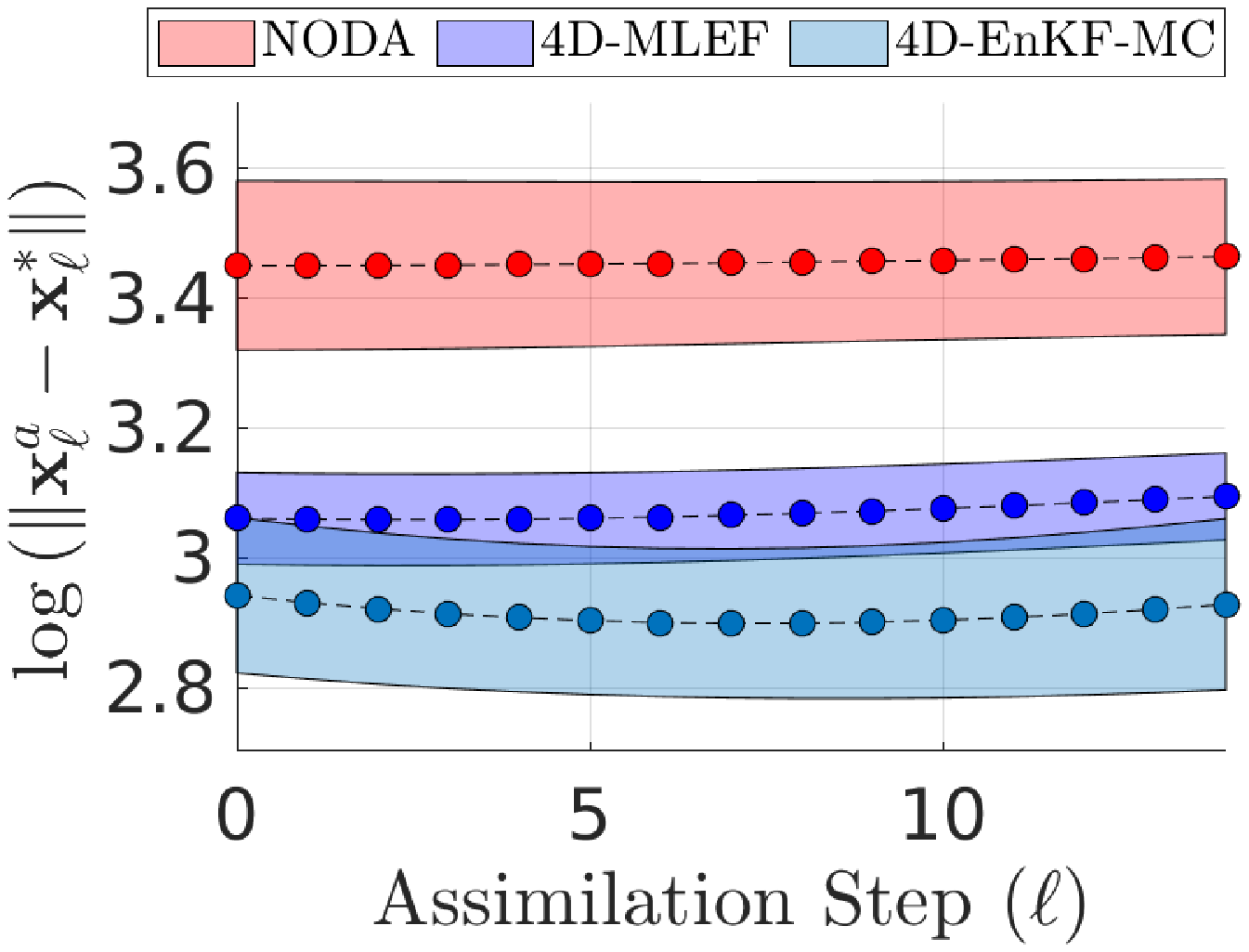}} \\ \hline
  \raisebox{-16mm}{\rotatebox[origin=c]{90}{$\gamma=6$}} &    \raisebox{-\totalheight}{\includegraphics[width=0.40\textwidth]{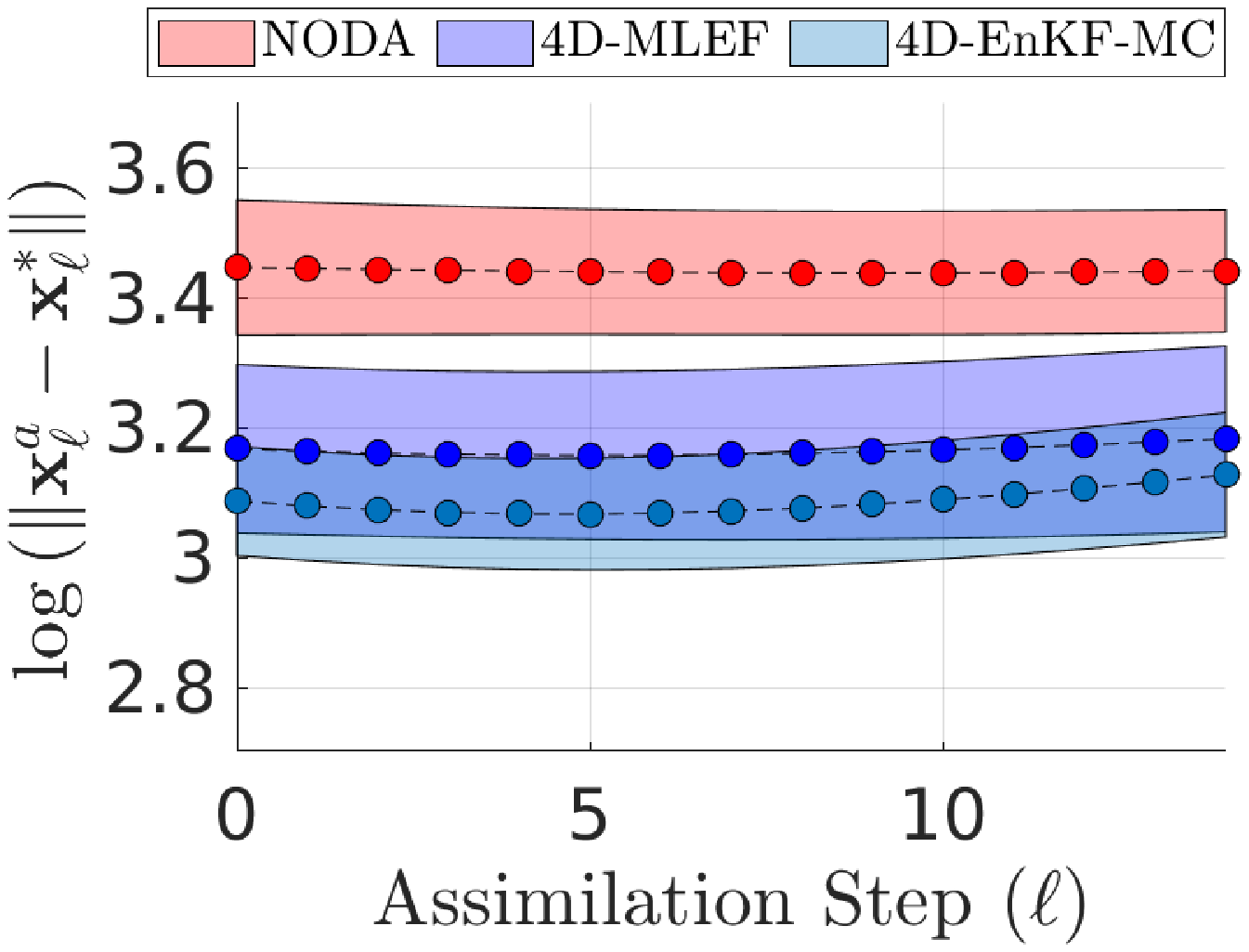}} &    \raisebox{-\totalheight}{\includegraphics[width=0.40\textwidth]{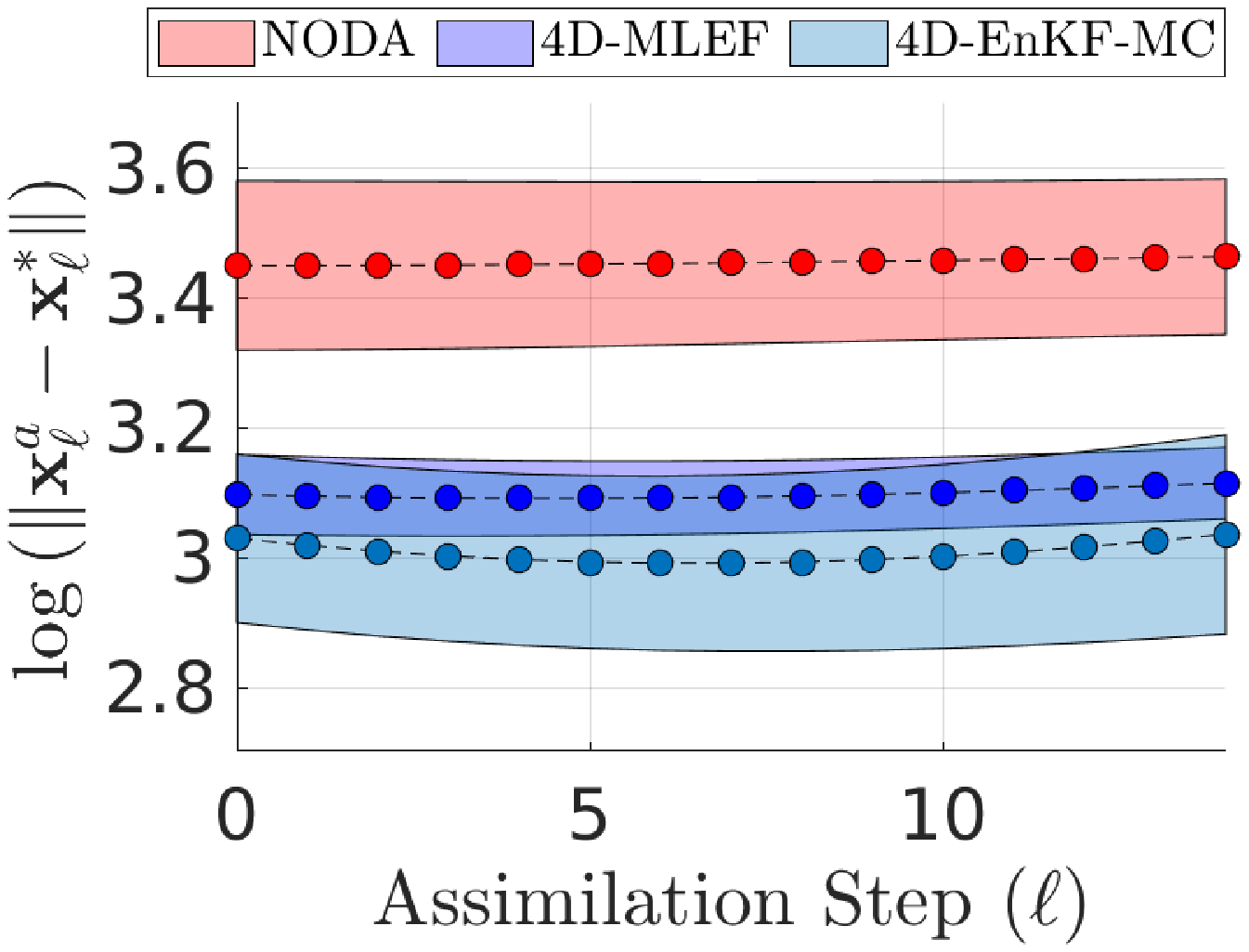}} \\ \hline
   \raisebox{-16mm}{\rotatebox[origin=c]{90}{$\gamma=7$}} &    \raisebox{-\totalheight}{\includegraphics[width=0.40\textwidth]{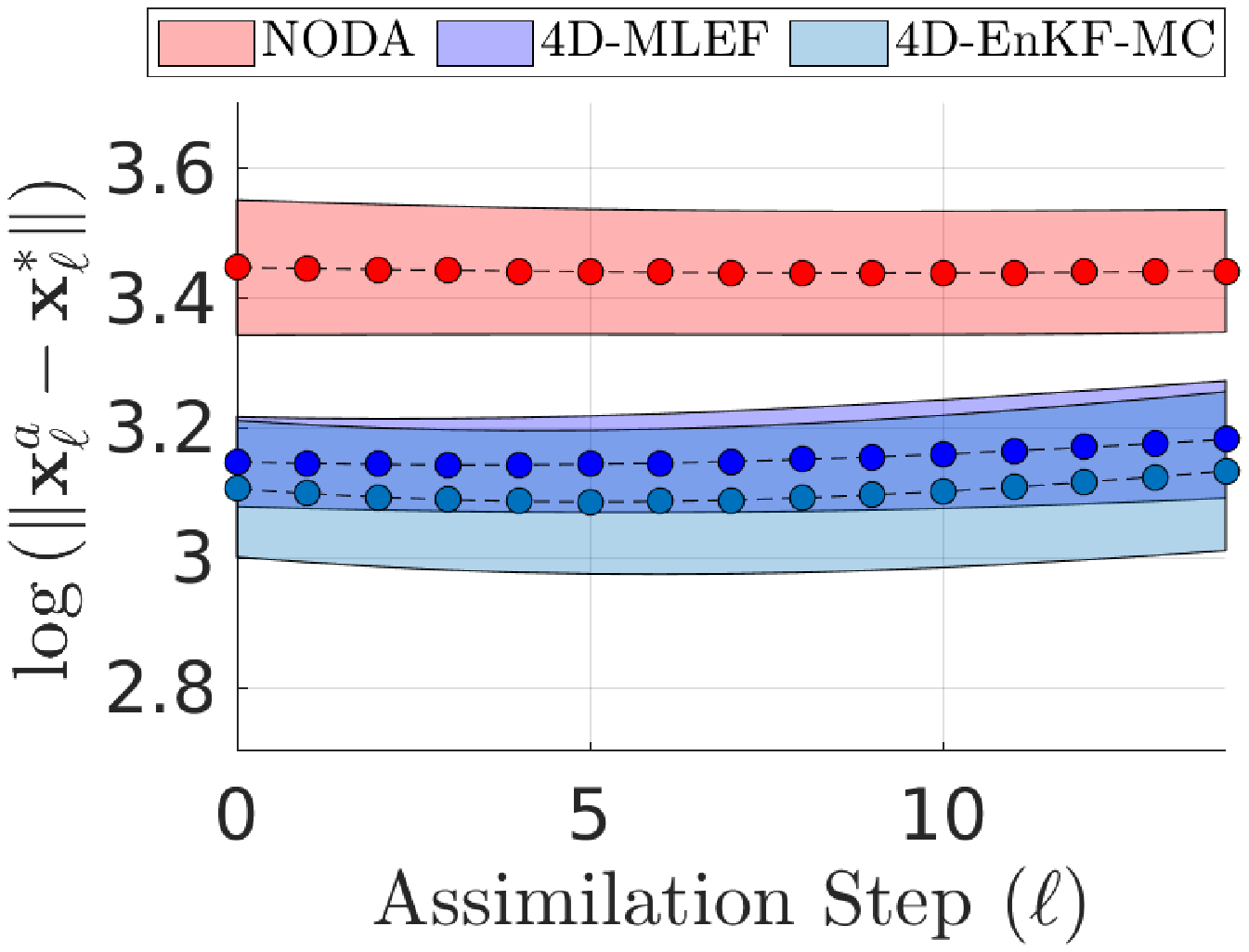}} &    \raisebox{-\totalheight}{\includegraphics[width=0.40\textwidth]{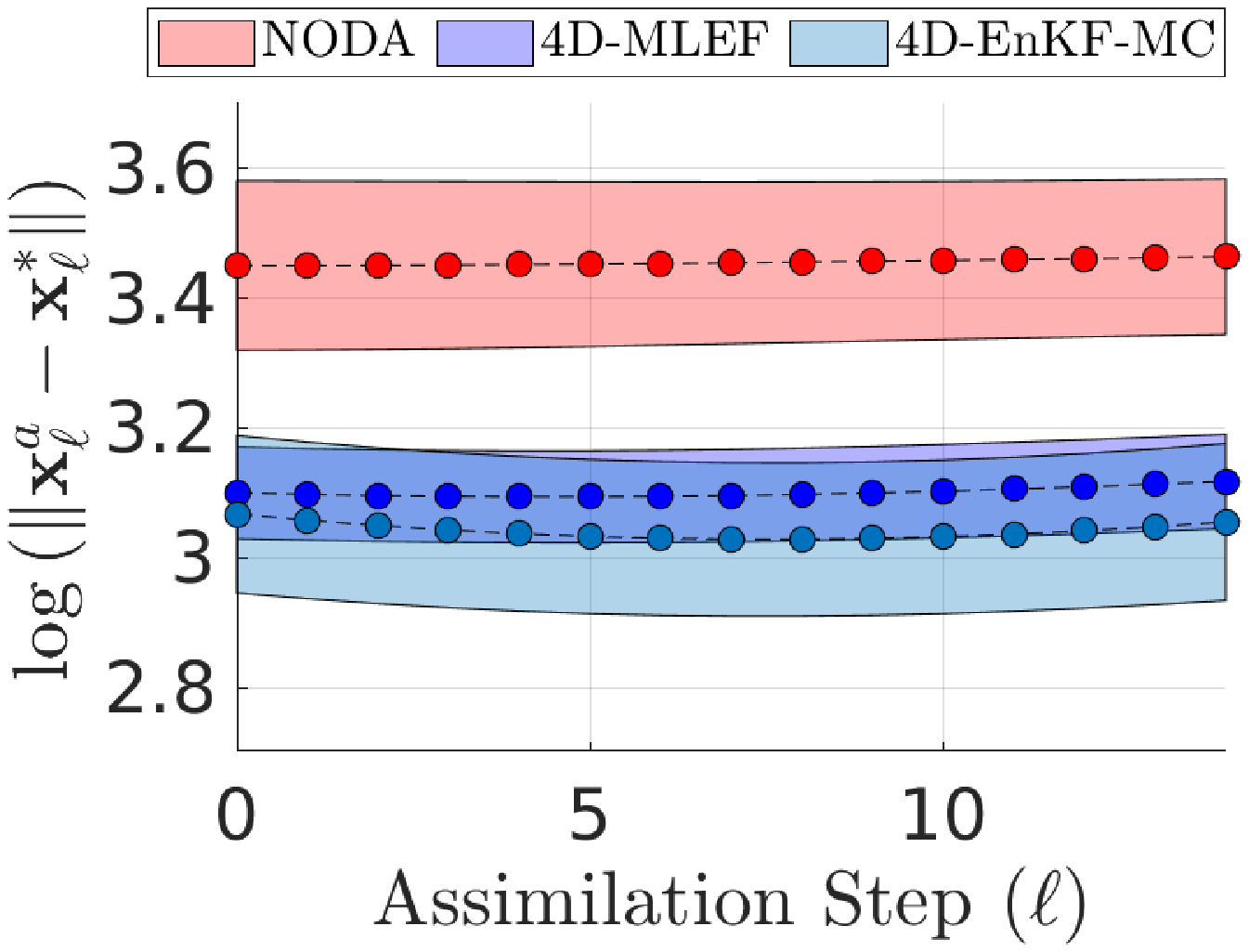}} \\ \hline
\end{tabular}
 \caption{Averages (dashed lines) and standard deviations (shaded regions) of error norms for 30 realizations of experiments and the values of parameters $p = 70\%, \Nens \in \{20,\,60\}$, and $4 \le \gamma \le 7$. The results are shown for the compared filter implementations as well as the NO Data Assimilation (NODA) forecast.}
\label{fig:4-7-70}
\end{figure}
\begin{figure}[H]
\centering
\begin{tabular}{ccc} \\ \hline
  & $\Nens = 20 \%$ & $\Nens = 60 \%$   \\ \hline
 \raisebox{-16mm}{\rotatebox[origin=c]{90}{$\gamma=4$}} &    \raisebox{-\totalheight}{\includegraphics[width=0.40\textwidth]{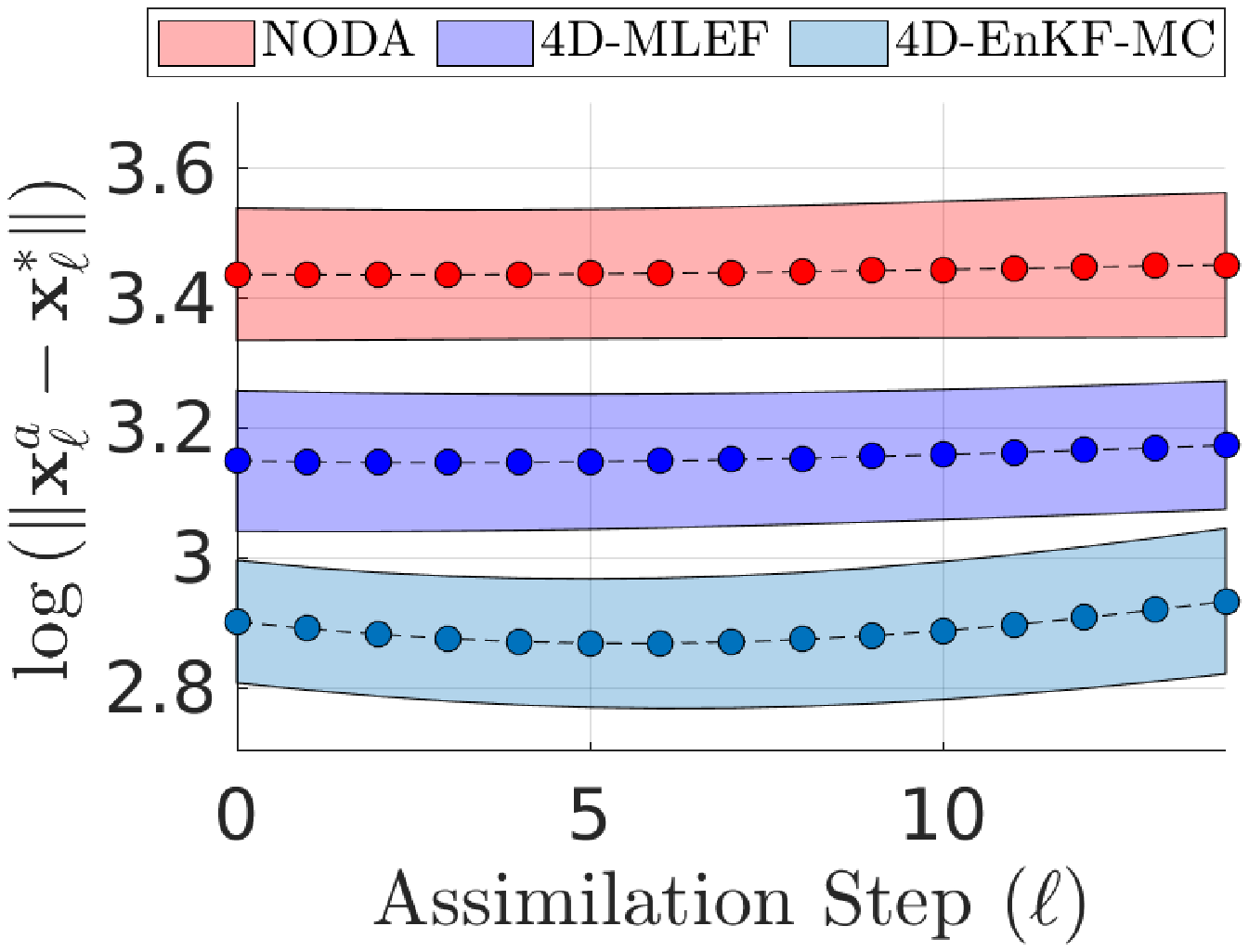}} &    \raisebox{-\totalheight}{\includegraphics[width=0.40\textwidth]{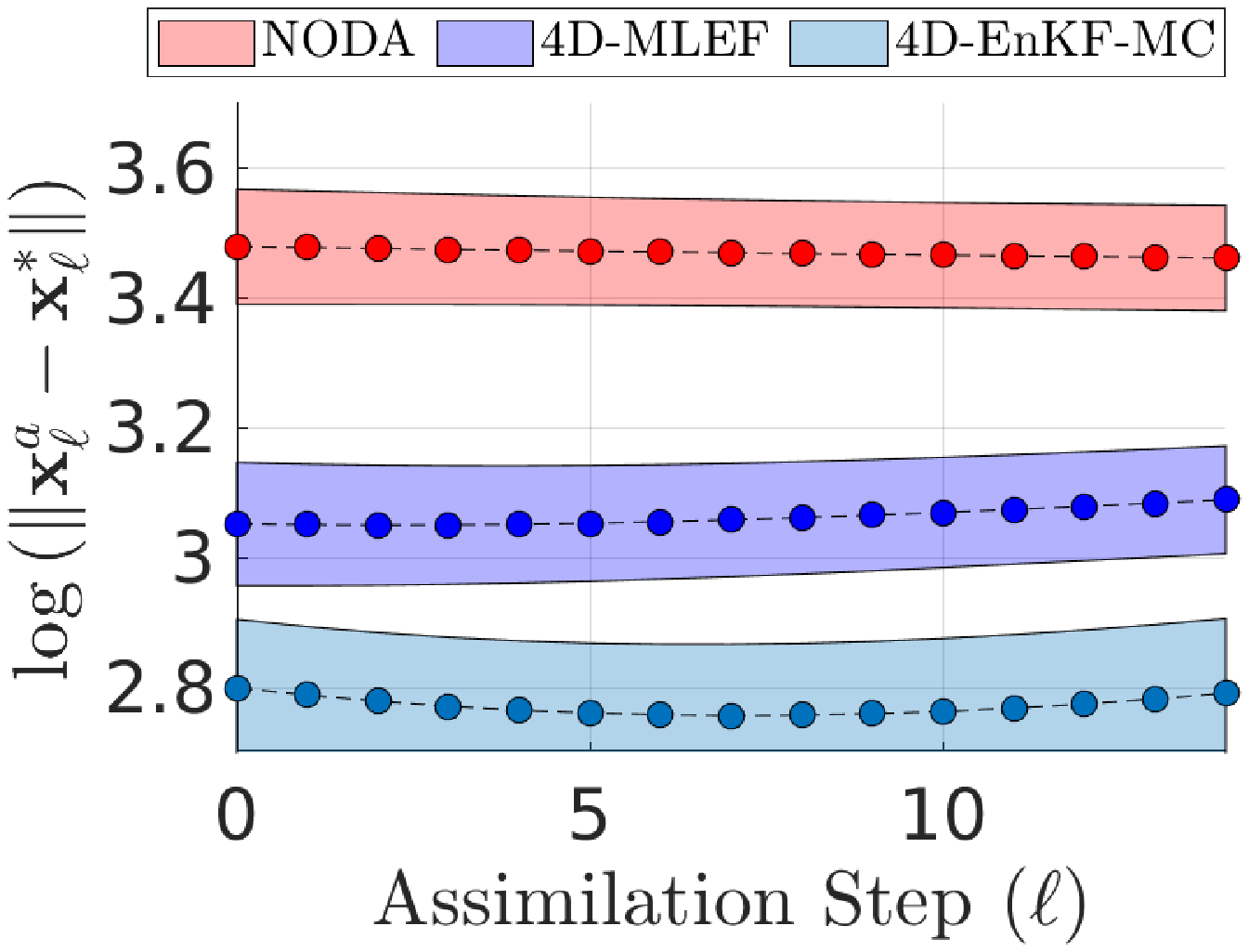}} \\ \hline
 \raisebox{-16mm}{\rotatebox[origin=c]{90}{$\gamma=5$}} &    \raisebox{-\totalheight}{\includegraphics[width=0.40\textwidth]{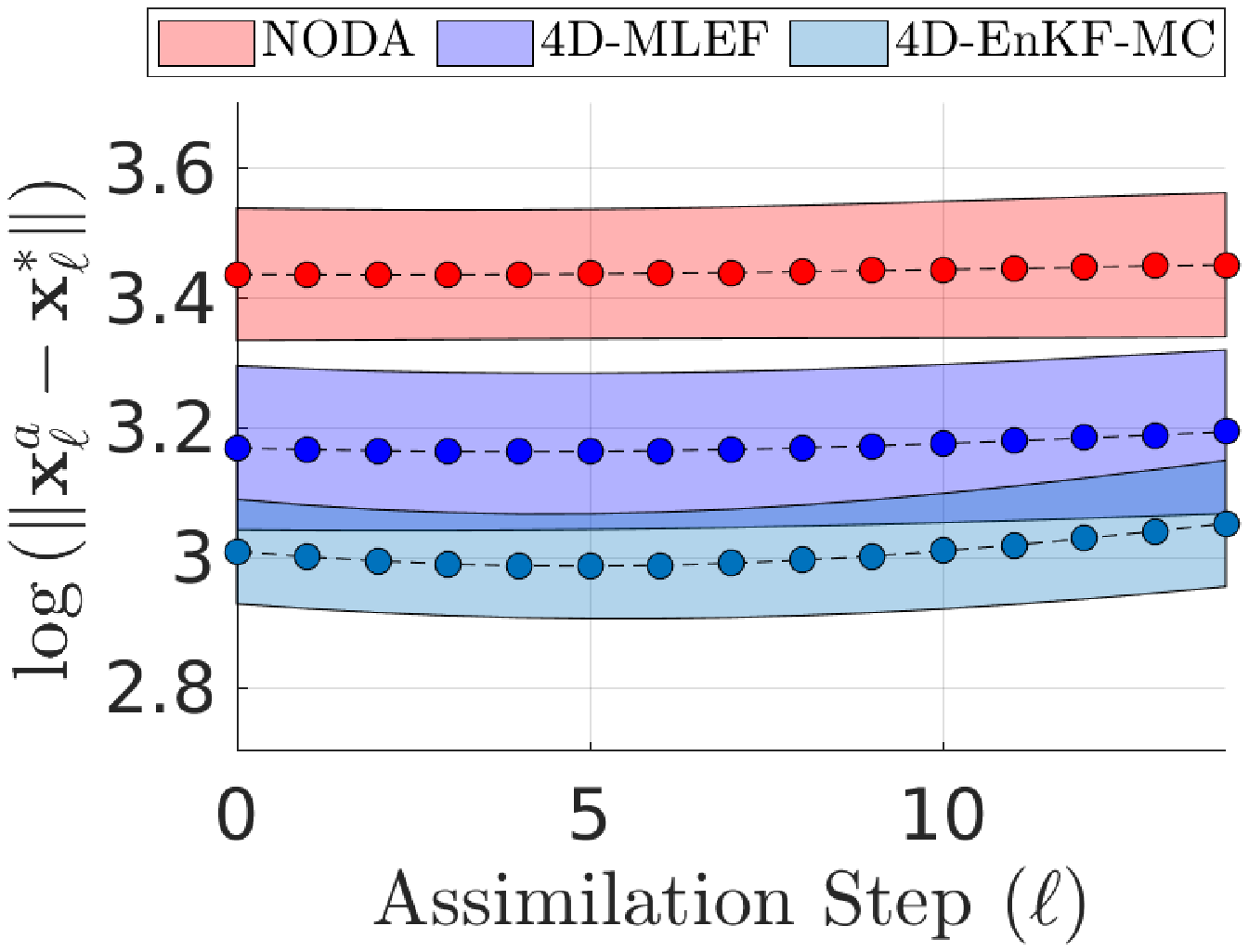}} &    \raisebox{-\totalheight}{\includegraphics[width=0.40\textwidth]{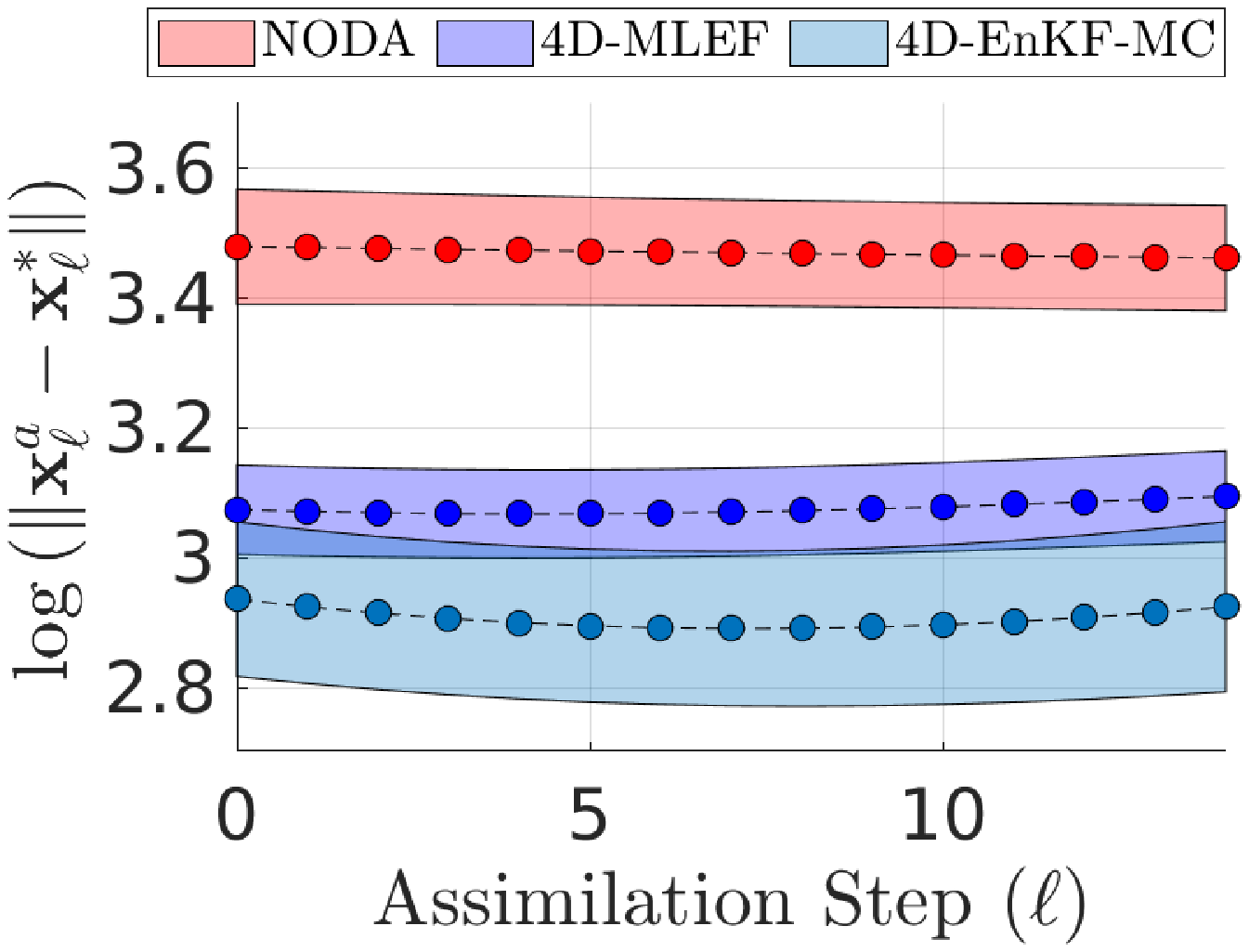}} \\ \hline
  \raisebox{-16mm}{\rotatebox[origin=c]{90}{$\gamma=6$}} &    \raisebox{-\totalheight}{\includegraphics[width=0.40\textwidth]{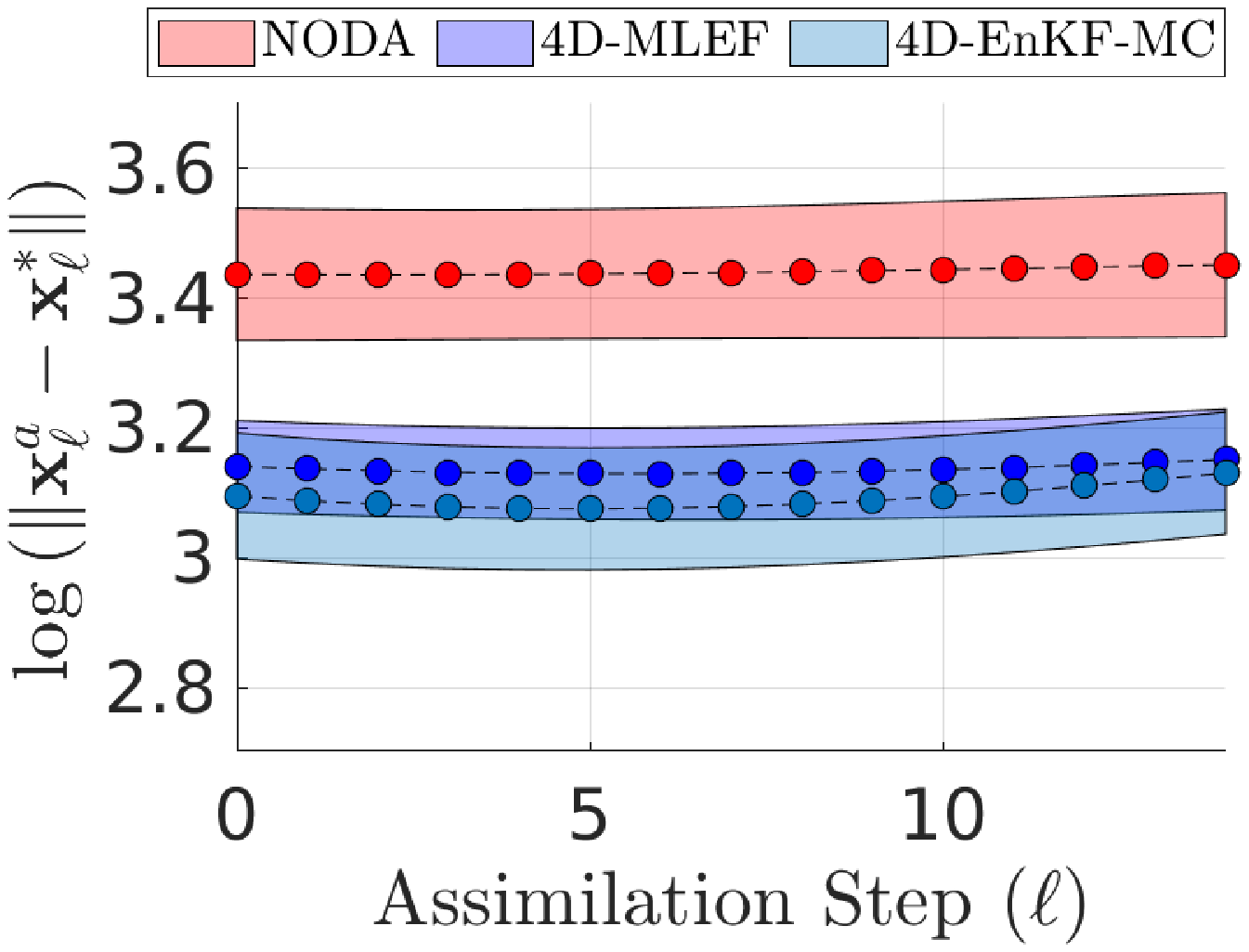}} &    \raisebox{-\totalheight}{\includegraphics[width=0.40\textwidth]{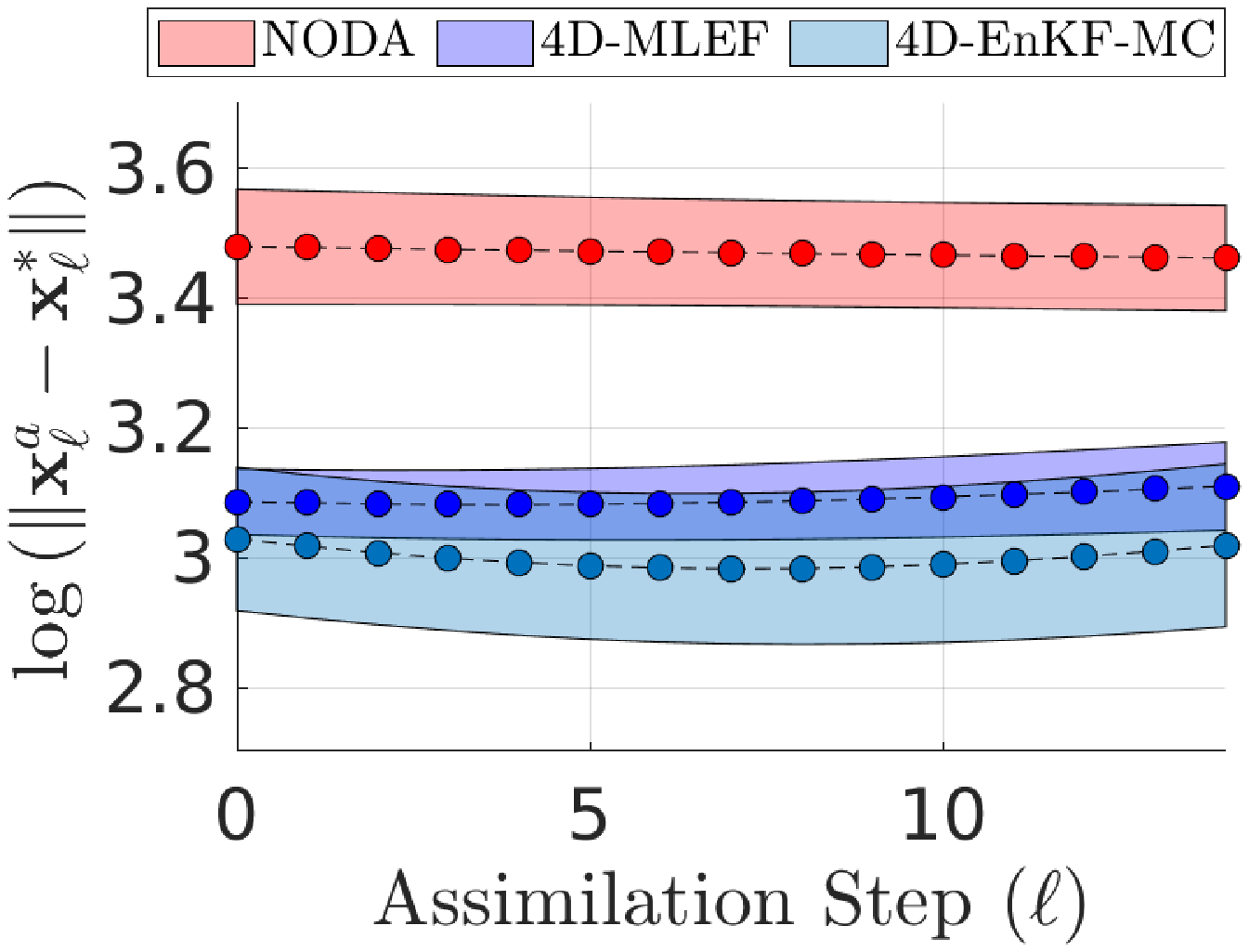}} \\ \hline
   \raisebox{-16mm}{\rotatebox[origin=c]{90}{$\gamma=7$}} &    \raisebox{-\totalheight}{\includegraphics[width=0.40\textwidth]{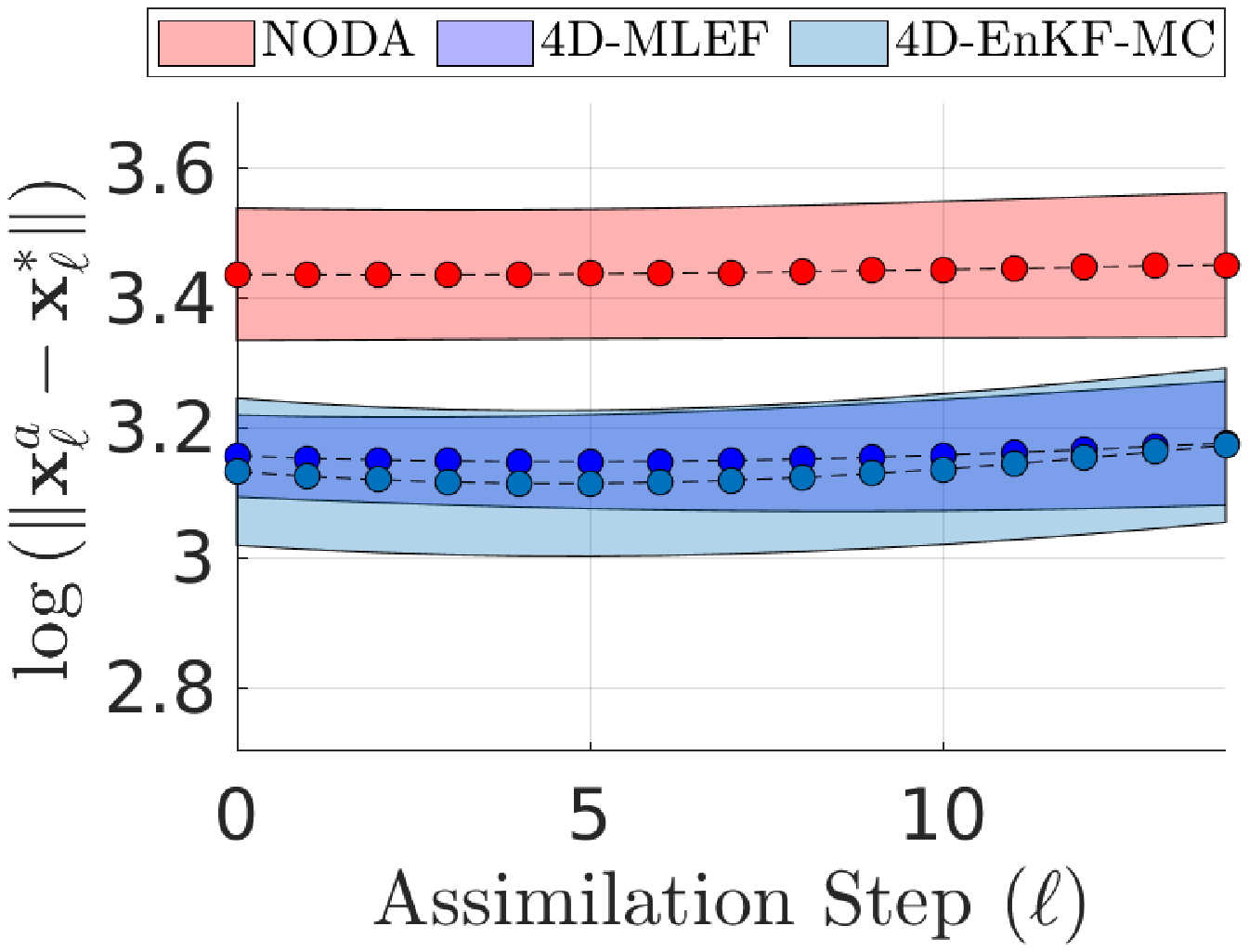}} &    \raisebox{-\totalheight}{\includegraphics[width=0.40\textwidth]{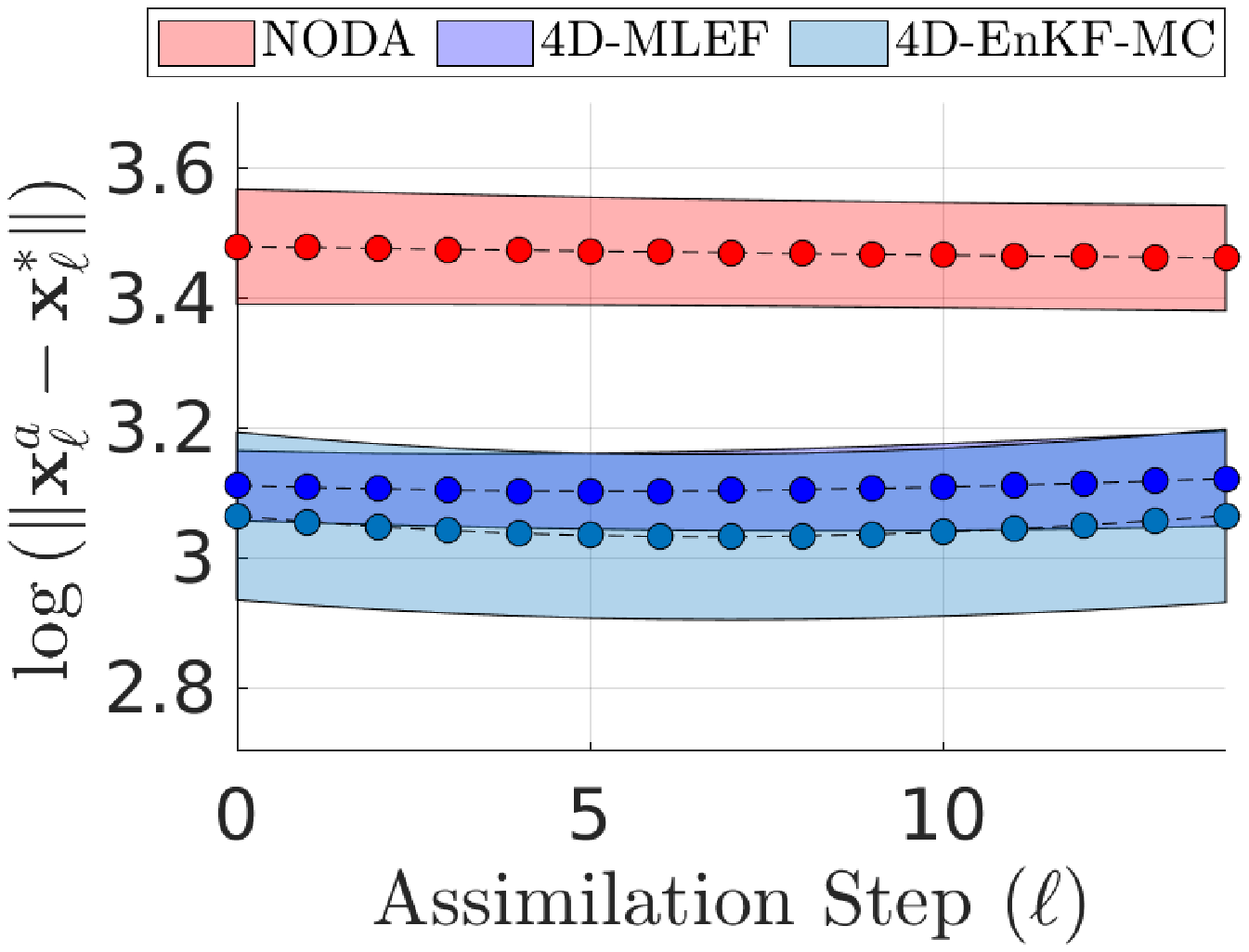}} \\ \hline
\end{tabular}
 \caption{Averages (dashed lines) and standard deviations (shaded regions) of error norms for 30 realizations of experiments and the values of parameters $p = 100\%, \Nens \in \{20,\,60\}$, and $4 \le \gamma \le 7$.}
\label{fig:4-7-100}
\end{figure}
\begin{table}[H]
\centering
\begin{tabular}{|c|c|c|c|c|c|} \hline
$\gamma$ & $\Nens$ & $\ra$ & 4D-EnKF-MC & 4D-MLEF & NODA \\ \hline
\multirow{6}{*}{1} & \multirow{3}{*}{20} &    2 & 0.158 & 22.397 & \multirow{6}{*}{31.328}  \\ \cline{3-5} & & 6 & 0.549 & 22.967 &   \\ \cline{3-5} & & 18 & 18.233 & 19.743 &   \\ \cline{2-5}& \multirow{3}{*}{60} &    2 & 0.144 & 22.206 &   \\ \cline{3-5} & & 6 & 0.173 & 22.167 &   \\ \cline{3-5} & & 18 & 0.629 & 22.697 &   \\ \cline{2-5} \cline{1-6}\multirow{6}{*}{2} & \multirow{3}{*}{20} &    2 & 0.276 & 22.944 & \multirow{6}{*}{31.387}  \\ \cline{3-5} & & 6 & 2.172 & 26.318 &   \\ \cline{3-5} & & 18 & 18.725 & 19.968 &   \\ \cline{2-5}& \multirow{3}{*}{60} &    2 & 0.148 & 22.609 &   \\ \cline{3-5} & & 6 & 0.234 & 23.261 &   \\ \cline{3-5} & & 18 & 1.665 & 25.636 &   \\ \cline{2-5} \cline{1-6}\multirow{6}{*}{3} & \multirow{3}{*}{20} &    2 & 11.230 & 23.642 & \multirow{6}{*}{31.387}  \\ \cline{3-5} & & 6 & 14.081 & 22.601 &   \\ \cline{3-5} & & 18 & 20.093 & 20.985 &   \\ \cline{2-5}& \multirow{3}{*}{60} &    2 & 7.730 & 22.077 &   \\ \cline{3-5} & & 6 & 10.534 & 21.661 &   \\ \cline{3-5} & & 18 & 12.979 & 21.499 &   \\ \cline{2-5} \cline{1-6}\multirow{6}{*}{5} & \multirow{3}{*}{20} &    2 & 20.736 & 24.138 & \multirow{6}{*}{31.387}  \\ \cline{3-5} & & 6 & 20.065 & 24.411 &   \\ \cline{3-5} & & 18 & 20.970 & 22.574 &   \\ \cline{2-5}& \multirow{3}{*}{60} &    2 & 19.225 & 21.930 &   \\ \cline{3-5} & & 6 & 18.550 & 21.596 &   \\ \cline{3-5} & & 18 & 18.736 & 21.772 &   \\ \cline{2-5} \cline{1-6}\multirow{6}{*}{6} & \multirow{3}{*}{20} &    2 & 22.629 & 24.236 & \multirow{6}{*}{31.387}  \\ \cline{3-5} & & 6 & 21.977 & 23.924 &   \\ \cline{3-5} & & 18 & 21.386 & 23.040 &   \\ \cline{2-5}& \multirow{3}{*}{60} &    2 & 20.540 & 22.365 &   \\ \cline{3-5} & & 6 & 20.449 & 22.200 &   \\ \cline{3-5} & & 18 & 20.927 & 21.826 &   \\ \cline{2-5} \cline{1-6}\multirow{6}{*}{7} & \multirow{3}{*}{20} &    2 & 23.209 & 23.837 & \multirow{6}{*}{31.387}  \\ \cline{3-5} & & 6 & 22.368 & 23.527 &   \\ \cline{3-5} & & 18 & 21.650 & 23.476 &   \\ \cline{2-5}& \multirow{3}{*}{60} &    2 & 20.947 & 22.432 &   \\ \cline{3-5} & & 6 & 21.088 & 22.275 &   \\ \cline{3-5} & & 18 & 22.512 & 22.551 &   \\ \cline{2-5} \cline{1-6}
\end{tabular}
 \caption{Averages of RMSE values for 30 realizations of experiments and the parameter values: $p = 70\%, \Nens \in \{20,\,60\}$, and $1 \le \gamma \le 7$.}
\label{tab:rho-0.7}
\end{table}
\begin{table}[H]
\centering
\begin{tabular}{|c|c|c|c|c|c|} \hline
$\gamma$ & $\Nens$ & $\ra$ & 4D-EnKF-MC & 4D-MLEF & NODA \\ \hline
\multirow{6}{*}{1} & \multirow{3}{*}{20} &    2 & 0.143 & 22.396 & \multirow{6}{*}{31.463}  \\ \cline{3-5} & & 6 & 0.520 & 22.948 &   \\ \cline{3-5} & & 18 & 17.963 & 19.680 &   \\ \cline{2-5}& \multirow{3}{*}{60} &    2 & 0.110 & 22.159 &   \\ \cline{3-5} & & 6 & 0.145 & 22.160 &   \\ \cline{3-5} & & 18 & 0.536 & 22.522 &   \\ \cline{2-5} \cline{1-6}\multirow{6}{*}{2} & \multirow{3}{*}{20} &    2 & 0.280 & 22.838 & \multirow{6}{*}{31.404}  \\ \cline{3-5} & & 6 & 1.890 & 24.636 &   \\ \cline{3-5} & & 18 & 18.591 & 19.748 &   \\ \cline{2-5}& \multirow{3}{*}{60} &    2 & 0.131 & 22.349 &   \\ \cline{3-5} & & 6 & 0.202 & 22.691 &   \\ \cline{3-5} & & 18 & 1.476 & 23.618 &   \\ \cline{2-5} \cline{1-6}\multirow{6}{*}{3} & \multirow{3}{*}{20} &    2 & 11.261 & 23.013 & \multirow{6}{*}{31.404}  \\ \cline{3-5} & & 6 & 14.209 & 22.643 &   \\ \cline{3-5} & & 18 & 20.046 & 20.998 &   \\ \cline{2-5}& \multirow{3}{*}{60} &    2 & 8.117 & 22.882 &   \\ \cline{3-5} & & 6 & 10.734 & 21.926 &   \\ \cline{3-5} & & 18 & 12.720 & 21.607 &   \\ \cline{2-5} \cline{1-6}\multirow{6}{*}{5} & \multirow{3}{*}{20} &    2 & 20.463 & 23.668 & \multirow{6}{*}{31.404}  \\ \cline{3-5} & & 6 & 20.310 & 24.062 &   \\ \cline{3-5} & & 18 & 21.019 & 22.643 &   \\ \cline{2-5}& \multirow{3}{*}{60} &    2 & 19.090 & 22.293 &   \\ \cline{3-5} & & 6 & 18.432 & 21.718 &   \\ \cline{3-5} & & 18 & 18.842 & 21.701 &   \\ \cline{2-5} \cline{1-6}\multirow{6}{*}{6} & \multirow{3}{*}{20} &    2 & 22.188 & 23.586 & \multirow{6}{*}{31.404}  \\ \cline{3-5} & & 6 & 22.128 & 23.080 &   \\ \cline{3-5} & & 18 & 21.343 & 23.024 &   \\ \cline{2-5}& \multirow{3}{*}{60} &    2 & 20.197 & 22.196 &   \\ \cline{3-5} & & 6 & 20.216 & 22.026 &   \\ \cline{3-5} & & 18 & 20.955 & 22.054 &   \\ \cline{2-5} \cline{1-6}\multirow{6}{*}{7} & \multirow{3}{*}{20} &    2 & 23.676 & 24.389 & \multirow{6}{*}{31.404}  \\ \cline{3-5} & & 6 & 23.089 & 23.566 &   \\ \cline{3-5} & & 18 & 21.626 & 23.394 &   \\ \cline{2-5}& \multirow{3}{*}{60} &    2 & 21.098 & 22.340 &   \\ \cline{3-5} & & 6 & 21.197 & 22.437 &   \\ \cline{3-5} & & 18 & 21.785 & 22.226 &   \\ \cline{2-5} \cline{1-6}
\end{tabular}
 \caption{Averages of RMSE values for 30 realizations of experiments and the parameter values: $p = 100\%, \Nens \in \{20,\,60\}$, and $1 \le \gamma \le 7$.}
 \label{tab:rho-1}
\end{table}

\section{Conclusions}
\label{sec:conclusions}

We propose two adjoint free 4D-Var methods for non-linear data assimilation. In both methods, control spaces are built to avoid the intrinsic need for adjoint models in the 4D-Var context. The techniques work as follows: based on an ensemble of model realizations, control spaces are built at observation times, then model states are constrained to the space spanned by those control spaces. The analysis increments are computed onto the proposed spaces by linearizing observation operators together with employing line-search optimization methods. The global convergence of the method can be proven as long as the control-space dimension equals that of the model one. We also propose a matrix-free method to compute search directions whose computational effort is linearly bounded by the model dimension; this feature is attractive under data assimilation operational settings. Experimental tests are performed by using the Lorenz-96 model and by employing a non-linear observation operator. The results reveal that both filters can compute initial posterior modes of the error distribution with reasonable accuracy in terms of RMSE values. Even more, the 4D-EnKF-MC can be employed to mitigate the impact of sampling noise as well as to guarantee the convergence of line-search methods during optimization steps.

\section*{Acknowledgments}
This work was supported in part by award UN 2018-38, and by the Applied Math and Computer Science Lab at Universidad del Norte, Colombia .

%
 \section*{Conflict of interest}

 The authors declare that they have no conflict of interest.

\bibliographystyle{spmpsci}      
\bibliography{Main}   


\end{document}